\documentclass[journal]{IEEEtran}
\usepackage{cite}
\usepackage{amsmath,amssymb,amsfonts}
\usepackage{graphicx}
\usepackage{textcomp}
\usepackage{xcolor}
\def\BibTeX{{\rm B\kern-.05em{\sc i\kern-.025em b}\kern-.08em
    T\kern-.1667em\lower.7ex\hbox{E}\kern-.125emX}}

\usepackage{amsthm}
\newtheorem{theorem}{Theorem}
\newtheorem*{remark}{Remark}
\newtheorem{lemma}{Lemma}
\newtheorem{corol}{Corollary}

\usepackage{tkz-kiviat}
\usetikzlibrary{arrows}
\newcommand\ColorBox[1]{\textcolor{#1}{\rule{2ex}{2ex}}}

\usepackage{algorithm}
\usepackage{algorithmicx}
\usepackage[noend]{algpseudocode}
\usepackage{comment}

\usepackage{tikz}
\usetikzlibrary{positioning}
\usetikzlibrary{arrows.meta}

\usepackage{booktabs} 
\usepackage{multirow}
\usepackage{enumitem}
\usepackage{flushend}

\begin{document}

\title{
Distributed Control of Charging for Electric Vehicle Fleets under Dynamic Transformer Ratings
}

\author{\IEEEauthorblockN{Micah Botkin-Levy},
\IEEEauthorblockN{Alexander Engelmann},
\IEEEauthorblockN{Tillmann M{\"u}hlpfordt},
\IEEEauthorblockN{Timm Faulwasser},
\IEEEauthorblockN{Mads Almassalkhi}

\thanks{This work was partly supported by the U.S. Department of Energy's Advanced Research Projects Agency - Energy (ARPA-E) award DE-AR0000694.  M. Botkin-Levy and M. Almassalkhi are with the Department of Electrical and Biomedical Engineering at the University of Vermont, Burlington, VT, USA. M. Almassalkhi is co-founder of startup company Packetized Energy, which is commercializing Packetized Energy Management.  Emails: \textit{\{mbotkinl, malmassa\}@uvm.edu}.
A. Engelmann and T. M\"uhlpfordt are with the Institute for Automation and Applied Informatics at the Karlsruhe Institute of Technology. Emails: \textit{\{alexander.engelmann, tillmann.muehlpfordt\}@kit.edu}.
T. Faulwasser is with the Institute for Energy Systems, Energy Efficiency and Energy Economics, Department of Electrical Engineering and Information Technology, TU Dortmund University, Dortmund, Germany. Previous affiliation:  Institute for Automation and Applied Informatics at the Karlsruhe Institute of Technology.  Email: \textit{timm.faulwasser@ieee.org}}.
}

\maketitle

\begin{abstract}
Due to their large power draws and increasing adoption rates, electric vehicles (EVs) will become a significant challenge for electric distribution grids. 
However, with proper charging control strategies, the challenge can be mitigated without the need for expensive  grid reinforcements. This manuscript presents and analyzes new distributed charging control methods
to coordinate EV charging under nonlinear transformer temperature ratings. Specifically, we assess the trade-offs between required data communications, computational efficiency, and optimality guarantees for different control strategies based on a convex relaxation of the underlying nonlinear transformer temperature dynamics. 
Classical distributed control methods such as those based on dual decomposition and alternating direction method of multipliers (ADMM) are compared against the new Augmented Lagrangian-based Alternating Direction Inexact Newton (ALADIN) method and a novel low-information, look-ahead version of packetized energy management (PEM). These algorithms are implemented and analyzed for two case studies on residential and commercial EV fleets. 
 Simulation results validate the new methods and provide insights into key trade-offs.

\end{abstract}


\section{Introduction}

  As renewable generation is increasingly deployed, powering our transportation system from the electric grid, instead of fossil fuels, will reduce emissions and climate change impacts. In addition, falling lithium-ion battery prices~\cite{hodges_2018} 
  and low maintenance costs~\cite{zart_2018} 
  will further increase adoption rates for both residential and commercial EVs. 
  However, there are certain challenges associated with the increased adoption of electric vehicles. 
  Specifically, uncoordinated charging from electric vehicles can lead to demand that exceeds the rating of distribution substation power transformers~\cite{osti_1430826}
  . Such MVA-scale transformer has its cores immersed in mineral oil for improved heat transfer. However, EVs will increase the loading on a transformer and result in a higher hot-spot temperature, which is the transformers highest internal temperature. The hot-spot temperature is a major factor in transformer wear-and-tear and aging as the hot oil will break down the winding insulation faster~\cite{xfrm1}.  To accurately model the transformer hot-spot temperature dynamics, a high-order, non-linear thermodynamic model, such as the IEEE Standard C57.91-1995 (e.g., \textit{Clause~7} and \textit{Annex~G}) is often used~\cite{ieeestd_annexG}. 
  
  Thus, it is desirable to manage the charging rate of EVs with respect to the transformer's hot-spot temperature limit and EV-specific objectives and constraints, which can be formulated as a multi-period scheduling problem. Due to a potentially large number of EVs and a long (over-night) prediction horizon, this scheduling problem can be computationally intensive and require full state information, which can raise data \textit{privacy} concerns. Techniques such as primal or dual decomposition are helpful in decoupling a large scheduling problem with coupling constraints into many smaller problems. Two classical algorithms for this purpose are dual decomposition and Alternating Direction Method of Multipliers (ADMM). In this manuscript, we present two novel distributed methods and compare them against the two classical methods
  in terms of how effectively they converge to a solution (i.e., \textit{processing}), the required data communications (i.e., \textit{privacy}), and optimality of the solution (i.e., \textit{performance}).

  For a general comparison of non-centralized control techniques, please see~\cite{distControl1}. There are numerous papers that utilize dual decomposition~\cite{dualEVC1,dualEVC3} and the ADMM approaches~\cite{admmEVC1,admmEVC2,admmEVC3,admmEVC4,admmEVC5,compEVC} to solve various EV charging problem formulations. Other works employ novel and creative approaches, such as~\cite{singleControl1,singleControl2}, which coordinate EV charging under static transformer and voltage constraints using dual decomposition with reactive power compensation~\cite{singleControl1} and a shrunken-primal-dual sub-gradient algorithm that achieves valley-filling (grid-centric) objectives~\cite{singleControl2}. In 
  ~\cite{gametheory, gametheory2, gametheory3}, the authors leverage game theoretic approaches for large-scale populations of EVs, where the average charging dynamics can be steered to a globally optimal solution with fast convergence on the order of 1-100 iterations depending on system parameters.
 
The above works focus on specific information scenarios, such as full information, shared (neighbors), or decentralized (non-shared) information. However, with increased interest in controlling EV charging comes a growing concern for protecting EV owners' information. This can be achieved by minimizing or eliminating the need to communicate information to a central coordinator and, instead, use peer-to-peer technologies to enable transactive energy trading~\cite{blockchain1, blockchain2}. 

 In this paper, we compare the \textit{privacy} offered by the classical algorithms of Dual Ascent and ADMM and new EV charging algorithms based on ALADIN and PEM. This comparison is based on protecting valuable customer information, such as personal travel schedules.
To the best of the authors' knowledge, there is limited work that develops and compares non-centralized EV charging algorithms subject to dynamic capacity constraints. For example,~\cite{compEVC} computes the optimal scheduling of EVs under static capacity constraints and compares the trade-off between the convergence speed and the amount of communication required. However, the study only  considers different combinations of two similar algorithms and neglects  the nonlinear transformer temperature dynamics. 
  
 Furthermore, while much of literature focuses on residential EV charging, fewer papers consider charging needs of fleets or hubs of commercial EVs, such as school buses or delivery trucks, which engender different charging models. One paper aggregates EVs and optimizes the lowest electricity charging cost solution under linearized power flow constraints~\cite{hub1}. Another studies time-of-use pricing for a parking garage of EVs~\cite{hub2}. Other works coordinate aggregated EVs for use as a virtual battery~\cite{evStorage1,evStorage2} or for frequency control~\cite{evFreq} without considering individual EVs or local grid constraints. Herein, we develop a new energy-based fleet charging model that incorporates charging requirements of the individual EVs.

  With this work, we build on the initial receding-horizon model predictive control (MPC) approach from~\cite{EVC}, but employ and analyze a convex relaxation of a practical nonlinear model for the transformer temperature dynamics and augment analysis with two novel, distributed EV charging schemes.  While most previous works on predictive electric vehicle charging (EVC) control focuses on one method for a specific setting, this manuscript also compares multiple distributed methods and studies the trade-offs between information sharing, performance, and computational processing requirements. Specifically, this paper leverages a distributed optimization method with quadratic convergence, the new augmented Lagrangian-based alternating direction inexact Newton (ALADIN) method~\cite{ALADIN} and it proposes a new iteration-free, packet-based coordination scheme, packetized energy management (PEM)~\cite{pemEVC,almassalkhi2018chapter}. These different methods have hitherto not been developed or analyzed for the EVC problem under dynamic coupling constraints. Note that prior work on PEM for EVs only considered static charging constraints and defined device priorities based on the charging constraint rather than the device's local energy state~\cite{pemEVC}. Thus, to incorporate the dynamic constraints within PEM, we first extend~\cite{almassalkhi2018chapter}'s device-driven, locally defined, energy-based prioritization scheme to incorporate EVs' desired states of charge and departure times. In addition, we augment the coordinator from~\cite{almassalkhi2018chapter} with a look-ahead, mixed-integer quadratically constrained quadratic program (MIQCQP) to account for the temperature dynamics and packet requests. Finally, we present both residential and commercial EV charging case studies to compare the role of information across different EVC methods. While the EVC problem is technically challenging, it is also of immediate practical relevance for EV fleet operators~\cite{ups}.
  
  In Section~\ref{sec:models}, we formulate the nonlinear, thermal transformer model and local EV user energy/power constraints and in Section~\ref{sec:pwl} present a convex reformulation, which is rigorously analyzed. Then, in Section~\ref{sec:distributed}, we develop two new, non-centralized EVC algorithms, namely ALADIN and PEM, and briefly discuss the practical considerations facing a utility or a third-party coordinator/aggregator. We present two case studies in Sections~\ref{sec:evc} and~\ref{sec:hub} to validate our methods against conventional methods from literature and to serve as a comparison in Section~\ref{sec:comp}. We conclude the paper with a summary of the paper and recommendations for future research directions in Section~\ref{sec:conc}.

\begin{table}[t]
    \centering
     \caption{EVC parameters}    
     \resizebox{\columnwidth}{!}{
     \begin{tabular}{l l c c}
     \toprule
     \textbf{Variable} & \textbf{Description} & \textbf{Domain} & \textbf{Units}\\
      \midrule
     \multicolumn{4}{c}{\textit{System-wide parameters}} \\
     \midrule
     $N$ & Number of EVs & $\mathbb{Z}_+$ & - \\
     $T_a(k)$ & Ambient temperature at time $k$ & $\mathbb{R}_{+}$ & $^\circ$C \\
     $i_d(k)$ & Secondary background current at time $k$  & $\mathbb{R}_{+}$ & kA\\
    $R$ & Primary-secondary voltage ratio  & $(0,1]$ & -\\
     $T^\text{max}$ & Transformer temperature limit & $\mathbb{R}_{+}$ & $^\circ$C \\
     $\gamma$ & Ohmic losses-to-temp & $\mathbb{R}_{+}$ &  $^{^\circ C}/_{A^2}$\\
      $\tau$ & Temp time constant & $\mathbb{R}_{+}$ & -\\
      $\rho$ & Ambient-to-temp losses & $\mathbb{R}_{+}$ & -\\
     $K$ & Optimization horizon length & $\mathbb{Z}_{++}$ &  \# of time steps\\
     $\Delta t$ & Time step length & $\mathbb{R}_{++}$ & Seconds \\
      \midrule
     \multicolumn{4}{c}{\textit{EV-specific parameters for EV $n$}} \\
     \midrule
     \rule{0pt}{3ex}    
     $i_n^\text{max}$ & Current limit  & $\mathbb{R}_{+}$ & A \\
     $\alpha_n$ & Charging efficiency & [0,1]  &-  \\
     $\beta_n$ & Battery capacity & $\mathbb{R}_{+}$  & J \\
     $\eta_n$ &  Normalized battery size & $\mathbb{R}_{+}$ &  $^{1}/_{A}$\\ 
     $\bar s_n$ & Minimum required SoC & [0,1] & -\\
     $\bar k_n$ & Latest time step to reach $\bar s_n$  & [0,$K$] &  -\\
     $q_n$,$r_n$ & Penalties on partial SoC, current draw & $\mathbb{R}_{+}$ & -\\
     \bottomrule
    \end{tabular}
    }
    \label{tbl:notation}
\end{table}

\section{Problem Formulation}\label{sec:models}


\begin{figure}[t]
    \centering
    \resizebox{\columnwidth}{!}{
       \includegraphics{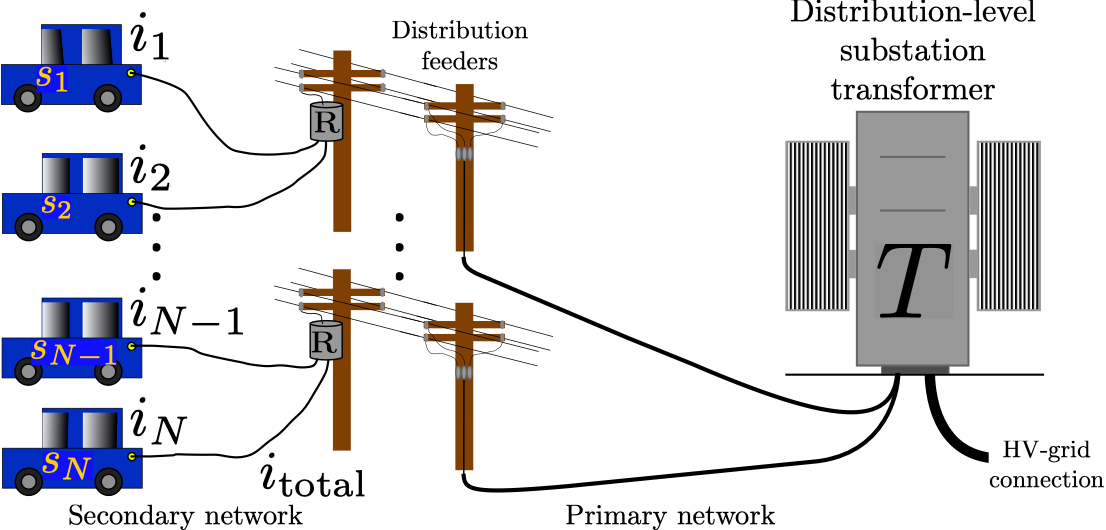}
        }
    \caption{Cartoon of the residential system setup, where the substation transformer's low-voltage (LV) side is in the primary network ($V_\text{pri}=8320$V) while the chargers reside in the secondary network ($V_\text{sec}=240$V).}
    \label{fig:sysModel}
\end{figure}

 Consider a finite collection of $N$ EVs with charging stations that are served by the same distribution-level substation transformer. Between a charger in the secondary network and the substation transformer in the primary network is a pole-top transformer, as shown in Fig.~\ref{fig:sysModel}. A dynamic transformer temperature model is used in the EVC formulation to keep the transformer hot-spot temperature below its limits while satisfying the local EV user constraints. The goal is to regulate the charging of all EVs within the transformer temperature limit. This gives rise to an MPC problem that is described at each time instance by a finite-horizon optimal control problem (OCP) of the following partially separable form: 
 \begin{align*}
     \min_{\textbf{u}(0),\hdots, \textbf{u}[K-1]} \: & \sum_{n=1}^N \sum_{k=0}^{K-1} q_n (x_n(k)-x_n^\text{ref}(k))^2  + r_n (u_n(k))^2 \\ 
     \end{align*}
     \begin{align*}
     \text{s.t. } \: & \mathbf{f}_k \left(\mathbf{x}(k+1), \mathbf{x}(k),\mathbf{u}(k)\right)=\mathbf{0} \\
     &\mathbf{h}_k \left(\mathbf{x}(k),\mathbf{u}(k)\right) \leq \mathbf{0} \\
    & x_n(k) \in \mathcal{X}_{n,k},
    \,\, u_n(k) \in \mathcal{U}_{n,k},
    \,\, x_n(0) = x_{\text{meas},n}, 
 \end{align*}
for $k=0,\hdots, K-1$ and $n=1,\hdots, N+1$, where $\mathbf{x}(k)\doteq [x_1(k),\hdots, x_N(k), x_{N+1}(k)]^\top \in  \mathbb{R}^{N+1}$ represents the states of charge (SoCs) for all $N$ EVs and the transformer temperature over the $K$ timesteps in the prediction horizon\footnote{In this work, the control and prediction horizons are assumed identical as the focus herein is on developing and comparing different novel algorithms.}. The control inputs $\mathbf{u}(k)\doteq [u_1(k), \hdots, u_N(k)]^\top\in \mathbb{R}^{N}$ include the EV charging rates. Functions $\mathbf{f}_k$ and $\mathbf{h}_k$ account for (coupling) inequality and equality constraints at time $k$, respectively, and are described in the next sections along with the objective function. 
The objective function's parameters $q_n\ge 0, r_n > 0$ represent the EV owner's preference for achieving the state reference value with minimal control effort. The compact, convex sets $\mathcal{X}_{n,k}$, $\mathcal{U}_{n,k}$ capture box constraints for states and inputs at time $k$.  
Table~\ref{tbl:notation} describes key parameters.

 \subsection{Transformer Dynamics and Constraints}
 
We consider the nonlinear hot-spot temperature model developed and validated in~\cite{paulXFRM}, 
\begin{equation}
    \dot T(t) = a L(t)^2 - b [T(t)-\tilde T_a(t)] + \tilde c, \label{xfrmCTdyn}
\end{equation}
where $T(t)$ represents the hot-spot temperature, $L(t)$ is the apparent power demand (volt-ampere or VA), and $ \tilde T_a(t)$ is the ambient temperature at time $t \in \mathbb{R}_+$. The  constant coefficients $a$, $b$, and $\tilde c$ represent the effects of conduction, convection, and radiation, respectively. These parameters may be estimated from experimental data (as done in~\cite{paulXFRM} with genetic programming) or from manufacturer spec sheets. In the present paper, the parameters are scaled versions of those in ~\cite{paulXFRM} such that the resulting model matches the timescale of the temperature responses given in spec sheets for the MVA-scale transformers used herein.


Using a zero-order hold with time-step $\Delta t$ for  the inputs and exact discretization, the discrete-time dynamics are 
\begin{align}
    T(k+1)=\tau T(k)+\tilde \gamma (L(k))^2 +\rho ( \tilde T_a(k)+c), \label{eq:tempDynLoad}
\end{align}
for $k=0,\hdots,K-1$ and measured initial temperature of $T(0)=T_\text{meas}$, where  $\tau=e^{-b \Delta t}$, $\rho=1-\tau$, $c=\frac{\tilde c}{b}$ and $\tilde \gamma=\frac{\rho}{b}$.

Since the control variables of interest are the EV charging currents, we will use a current-based model instead of a power-based model\footnote{Since the focus herein is on EV scheduling algorithms, the power system details associated with multi-phase distribution feeders and transformers, voltage fluctuations, and power factors are not discussed.}. Thus, we decompose $L(k)\doteq i^\text{pri}_\text{total}(k) V_\text{pri}$, where $i^\text{pri}_\text{total}(k)$ is the total RMS current magnitude from the primary network side of the transformer at time $k$ and supplied at fixed, rated RMS voltage $V_\text{pri}$.  Since the EV charger is supplied from the secondary network, its RMS voltage rating is $V_\text{sec} \doteq R V_\text{pri}$, where $R\in(0,1]$ is the pole-top transformer's fixed voltage ratio. The current $i^\text{pri}_\text{total}(k)$ is then the total \textit{reflected} current from the secondary network, i.e., $i_\text{total}^\text{pri}(k) = R i_\text{total}(k)$, where  $i_\text{total}(k)$ is composed of the background demand, $i_\text{d}(k)$, and all EV charging currents, $i_n(k)$, in the secondary network and
\begin{align*}
    i_\text{total}(k) = i_d(k) + \sum_{n=1}^N i_n(k).
\end{align*}
Thus, we can rewrite~\eqref{eq:tempDynLoad} in terms of   $i_\text{total}$ as
\begin{equation}
    T(k+1)=\tau T(k)+\gamma (i_\text{total}(k))^2 +\rho (T_a(k)), \label{xfrmDTdyn}
\end{equation} 
for $k=0,\hdots,K-1$ where $\gamma =\tilde \gamma V_{\text{pri}}^2R^2$ and $T_a(k)\doteq\tilde T_a(k)+\frac{\tilde c}{b}$.
In addition, the temperature $T(k)$ is constrained by the hot spot temperature limit $T^\text{max}$.

\subsection{EV Dynamics and Constraints}
The continuous-time, normalized charging dynamics of vehicle $n$ with current $i_n(t)$ are modeled as
\begin{equation}
    \dot s_n(t)= \Tilde{\eta}_n i_n(t),\quad s_n(t) \in [0,1],
\end{equation}
 where $\Tilde{\eta}_n\doteq \frac{\alpha_n}{\beta_n}V_{\text{sec}}$ is the normalizing ratio of the vehicle's charging efficiency ($\alpha_n$) to battery capacity ($\beta_n$) and supplied secondary RMS voltage, $V_{\text{sec}}$. The discrete time equation is
\begin{equation}
  s_n(k+1)=s_n(k)+\eta_n i_n(k), \quad s_n(0)=s_{\text{meas},n},
\end{equation}
for $k=0,\hdots,K-1$ and measured initial state of charge $s_{\text{meas},n}$. Each charger has a maximum current, $i_n^\text{max}$.

\subsection{EV Owner Preferences}
All vehicles are assumed to be available for charging at the beginning of the time period considered and owners have varying requirements for when they need their vehicle. The owners of the devices determine a minimum state of charge ($\bar s_n$) that must be met by a specific time step ($\bar k_n $). The associated constraint for the $n^{th}$ vehicle is
\begin{equation}
    s_n(k+1)\geq \hat s_n(k+1) \doteq
    \begin{cases} \bar s_n &  k+1\geq \bar k_n \\
                    0 & \text{else}
    \end{cases}.
\end{equation}

In addition, the user can set their preference for the trade-off between charging their EV quickly and minimizing local battery wear and control effort. This is achieved by selecting parameters in the objective function and is described next.

\subsection{EVC Control Objective}

 The $n^{th}$ EV owner's charging preference is used in the objective function to prioritize charging rate against the state of charge as
\begin{equation}
    J_n(\mathbf{i}_n,\mathbf{s}_n)\doteq \sum_{k=0}^{K-1} q_n(s_n(k+1)-1)^2 + r_n(i_n(k))^2 .
\end{equation}
Specifically, for each vehicle $n$, we define $M_n \doteq \frac{q_n}{r_n}\eta_n^2$ based on a user-defined ratio $\frac{q_n}{r_n}$, and fixed EV parameter $\eta_n>0$. This ratio $M_n$ will be used in the next section to provide sufficient conditions under which we can guarantee that a suitable convex relaxation is tight. These conditions are necessary due to the fundamental tradeoff in the objective function between reaching full charge quickly (large $M_n$ to maximize SoC) and keeping the battery charge-rate low (small $M_n$ to minimize control effort), which serves as a proxy for wear and tear.  
  This objective function is similar to a linear quadratic regulator (LQR) that penalizes deviations in SoC from unity and large control efforts. Summing over all $N$ vehicles yields the total cost metric, which we seek minimize in the optimization problem.


\subsection{Centralized Optimal Control Problem}

The open-loop OCP arises from the combination of the above constraints and objective function for all EVs and time steps. It reads
\begin{subequations} \label{eq:OCP}
    \begin{align}
\min_{i_n(k)}\;\;& \sum_{n=1}^N \sum_{k=0}^{K-1} q_n (s_n(k+1)-1)^2  + r_n (i_n(k))^2  \label{eq:objFun}\\
    \text{s.t.}~
    &T(k+1)=\tau T(k)+\gamma (i_\text{total}(k))^2 +\rho T_a(k) \label{cNLc} \\
    &s_n(k+1)=s_n(k)+\eta_n i_n(k) \label{cNLsn}\\
    &i_\text{total}(k) = i_d(k)+\sum_{n=1}^N i_n(k) \label{cNLcoup}\\
    &T(k+1)\leq T^\text{max}, \label{cPWLtmax} \\
    & s_n(k+1)\in[\hat s_n(k+1),~1] \label{cNLsnmin}\\
    & i_n(k)\in[0,~i_n^\text{max}] \label{cNLin} \\
    & T(0)=T_{\text{meas}},~s_n(0)=s_{\text{meas},n} 
    \end{align}
\end{subequations}
for all $k=0,\hdots, K-1$ and $n=1,\hdots, N$. This is a non-convex Nonlinear Program (NLP) due to the nonlinear~\eqref{cNLc}. 
Note that the only coupling constraint between the transformer and EV dynamics is~\eqref{cNLcoup}. Previous work in~\cite{EVC} used a linearized temperature model to simplify the coupling.

\section{Convexification of Centralized EVC Problem} \label{sec:pwl}
To overcome the non-convexity of~\eqref{cNLc}, we consider two different relaxations: an epigraph relaxation, which yields a second-order cone program (SOCP), and a piecewise linear (PWL) relaxation. The former  replaces the quadratic equality~\eqref{cNLc} with the linear equality and quadratic inequality
\begin{align}
    T(k+1)&=\tau T(k)+\gamma e(k)+\rho T_a(k) \label{eq:tempRelax}\\
    e(k)&\geq (i_\text{total}(k))^2, \label{qcqpRelax}.
\end{align}
 Under this relaxation, problem~\eqref{eq:OCP} becomes a SOCP. The benefit of the this approach is that
 if~\eqref{cPWLtmax} is strictly active at time $k$ then~\eqref{qcqpRelax} is satisfied with equality for all prior time steps and we recover the nonlinear model exactly.
 This is guaranteed by the following theorem and corollary.
\begin{theorem}[Main Result]\label{thm1}
Given fixed EV parameters $r_n\ge0$, $\eta_n, q_n >0$. If, at optimality, $\exists n, k$ for which $i_n(k)<i_n^\text{max}$ (i.e., an EV charger is throttled) and SoC satisfies
$$ 
s_n(k+1) < \left \{
\begin{array}{rr}
1 & \text{if $r_n=0$}\\
\frac{M_n+s_n(0)}{M_n +1} & \text{if $r_n>0$}
\end{array}
\right.
, 
$$ 
then $ e(l)=(i_\text{total}(l))^2 \,\, \forall l\le k$ in \eqref{qcqpRelax}. 
\end{theorem}

The proof is based on KKT analysis and is provided in the appendix. 
Note that when $r_n>0$, Theorem~\ref{thm1} provides a method to choose $q_n$ and $r_n$ based on constant $\eta_n$ and a desirable upper threshold on state of charge. Ideally, one would chose a threshold of $1$, but that requires $r_n=0$, which may not be reasonable. Instead, one could solve for $M_n$ by setting $M_n/(M_n+1) > \overline{s}_n$ (ignoring the initial state, $s_n(0)$), which then neatly embeds the user-defined QoS constraint into the objective function parameters. For example, if $\overline{s}_n = 0.8$, one can choose $M_n>4$, which implies $q_n/r_n > \frac{4}{\eta_n^2}$.

\begin{remark}[Tightness of the SOCP relaxation]
  At optimality, it may not be the case that any EV $n$ satisfies Theorem~\ref{thm1}'s conditions: $i_n(k) < i_n^\text{max}$ and $s_n(k+1) < \frac{M_n+s_n(0)}{M_n+1}$ for some timestep $k$. That is, the optimal solution may not be tight, if for all EVs $n$ and for entire prediction horizon $k$ either 
\begin{enumerate}[label=\Roman*.]
\item  $i_n(k)= i_n^\text{max}$ or
\item  $s_n(k+1) \ge \frac{M_n+s_n(0)}{M_n+1}$
\end{enumerate} 
In case~I, EVs are all charging at their maximum charge rates and not throttled, which indicates under-utilized capacity from the transformer. For case~II, the trade-offs from the objective function imply that any EVs that may be throttled must have a sufficiently high state of charge and are not negatively impacted by the transformer's capacity. Together,~I and~II imply that~\eqref{cPWLtmax} may not be strictly active, so the temperature state in~\eqref{eq:tempRelax} and the convex relaxation~\eqref{qcqpRelax} can be removed without affecting the optimal solution. Thus, outside of Theorem~\ref{thm1}'s conditions, the convex relaxation has no impact on the optimal solution, which ensures that no feasible solution for the relaxed SOCP formulation will lead to overheating of the transformer.
\end{remark}


Finally, to relate the transformer's temperature state and safety limit~\eqref{cPWLtmax} to the tight convex relaxation above, we present the following corollary. Together with Theorem~\ref{thm1}, this corollary guarantees that if the temperature limit~\eqref{cPWLtmax} is strictly active at time $k+1$ then the convex relaxation is tight for all prior timesteps.

\begin{corol}[Temperature limit]\label{corol1}
For the SOCP, at optimality, $k+1$ is the last instance for which~\eqref{cPWLtmax} is strictly active, if and only if, $k$ is the largest integer for which~\eqref{qcqpRelax} is tight.
\end{corol}

 Despite the guarantee of tightness for the relaxed model at optimality, the quadratic constraints increase the complexity of complementary conditions and begets numerical difficulties. To overcome this challenge, a piecewise linear (PWL) approach is used to formulate the nonlinear problem as a quadratic program (QP), which improves numerics of the problem significantly. An additional benefit of the PWL approximation is that the PWL segments dominate the quadratic model and, thus, is designed to over-estimate the transformer current as shown in Fig.~\ref{fig:Relax}. This over-estimate is a function of the number of segments and creates a conservative prediction of the transformer temperature.  Therefore, for the remainder of this manuscript we focus on the PWL implementation.

\subsection{Piecewise Linear Approximation}
 Define $e(k)$ as a PWL approximation of $i_\text{total}(k)^2$ with $M$ segments of equal width $\Delta i \doteq\frac{I^\text{max}}{M}$ as seen in Fig.~\ref{fig:Relax}, where $I^\text{max}$ is an upper bound on transformer current, then
 \begin{equation} \label{eq:PWL}
     i_\text{total}(k)^2 \le \text{PWL} \{i_\text{total}(k)^2\}=: e(k) =\sum_{m=1}^M\alpha_m i_m^\text{PW}(k),
 \end{equation}
where $i_m^\text{PW}(k) \in[0,\Delta i]$ represent auxillary PWL variables for each segment $m$ a time $k$ such that $\sum_{m=1}^M i_m^\text{PW}(k) \doteq i_\text{total}(k)$ and slope parameters $\alpha_m \doteq (2m-1)\Delta i$.

Note that this PWL approximation relaxes the adjacency conditions\footnote{Adjacency conditions enforce $i_m^\text{PW}(k) > 0 \Rightarrow i_p^\text{PW}(k) = \Delta i, \, \forall p < m$} that are usually enforced for the PWL segments, which avoids a mixed-integer formulation and creates the blue convex relaxation shown in Fig.~\ref{fig:Relax}. Using this directly in the transformer constraint turns the NLP into a relaxed QP:
\begin{equation}
    T(k+1)=\tau T(k)+\gamma \left( \Delta i \sum_{m=1}^M(2m-1)i_m^\text{PW}(k) \right) +\rho T_a(k). \label{cPWLt}
\end{equation}

\begin{figure}
    \centering
    \begin{tikzpicture}[>={Stealth[inset=1pt,length=8pt,angle'=40,round]},scale=0.75]
        \coordinate (O) at (0,0);
        \node[] (fLab) at (7,5) {$i^2(k)$};
        \draw[->] (0,0) -- (8,0) coordinate[label = {below:$i_\text{total}(k)$}] (xmax);
        \draw[->] (0,0) -- (0,6) coordinate[label = {left:$e(k)$}] (ymax);
        \draw[->,very thick,scale=8,domain=0:0.8,smooth,variable=\x,black,thick] plot ({\x},{\x*\x}); 
        \draw[very thick,scale=8,draw=blue!200,fill=blue!15,opacity=0.8] (0,0) -- node[above,] {$\alpha_1$}  (0.25,0.25*0.25) --  node[above] {$\alpha_2$}(0.5,0.5*0.5)--  node[above] {$\alpha_3$} (0.75,0.75*0.75)--(.5,0.5)-- node[below right,blue]{PWL}(0.25,5/16)--(0,0);
        \draw[fill=gray!60,scale=8, opacity=0.2] plot[smooth, samples=100, domain=0:.85] (\x,\x*\x) -|node[below right,black,opacity=0.75] {SOC} (0,0) --  cycle;
        \draw[shift={(0,0)}] (0pt,2pt) -- (0pt,-2pt)  node[below] {0};
        \draw[shift={(2,0)}] (0pt,2pt) -- (0pt,-2pt)  node[below] {$\Delta i$};
        \draw[shift={(4,0)}] (0pt,2pt) -- (0pt,-2pt)  node[below] {$2\Delta i$};
        \draw[shift={(6,0)}] (0pt,2pt) -- (0pt,-2pt)  node[below] {$3\Delta i$};
    \end{tikzpicture}
\caption{Relaxing the non-convexity $e(k) = (i_\text{total})^2$ with a PWL approximation that does not enforce adjacency conditions (blue) and a conic relaxation (gray). Note that the PWL approximation assumes that $ I^\text{max} = 3\Delta i$.} \label{fig:Relax}
\end{figure}
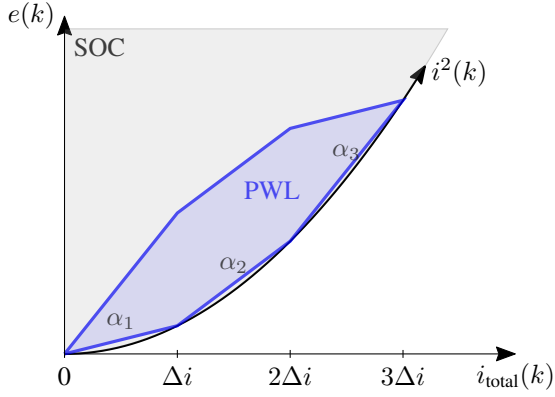

\begin{remark}[Upper bound on PWL error] Since we are using equal width segments, the maximum error between the PWL approximation and the actual $i^2$ is just the maximum distance between the linear segment ($\text{PWL}(i)$) and the quadratic curve ($q(i)$) at the midpoint (i.e. $\frac{\Delta i}{2}\doteq\frac{I^\text{max}}{2M}$), 
\begin{align}
    \epsilon_i^\text{max}&= \text{PWL}\left(\frac{\Delta i}{2}\right) - q\left(\frac{\Delta i}{2}\right) =\frac{(I^\text{max})^2}{2M^2} -\left(\frac{I^\text{max}}{2M}\right)^2 \\
    & \Rightarrow \epsilon_i^{max} =\frac{(I^\text{max})^2}{4M^2}.
\end{align}
Multiplying by $\gamma$ provides the upper bound on the corresponding temperature error:
\begin{equation}
    \epsilon_T^\text{max}=\frac{\gamma (I^\text{max})^2}{4M^2}.
\end{equation}
Even for a large current $I^\text{max}=0.72 kA$ with $\gamma=15.74 {^{^\circ C}/_{(\text{kA})^2}}$, and $M=6$ segments, the maximum error between a PWL's linear prediction of the transformer temperature ($T_\text{PWL}$) and the quadratic temperature ($T_q$) for a single time step is $\epsilon_T^{max}\doteq T_\text{PWL}(k+1) - T_q(k+1)= 0.057^\circ C$ when the convex relaxation is tight. While this temperature error accumulates over time steps in the open loop prediction, it is also discounted over time since $\tau<1$. Therefore, the piecewise linear approximation provides a feasible and robust estimate of the nonlinear temperature dynamics.  
\end{remark}

\subsection{Centralized PWL Problem}
The PWL relaxation provides an approximation of the transformer dynamics in~\eqref{cPWLt} and replaces~\eqref{cNLc}. In addition, the current coupling constraints between the PWL current segments and the EV currents are:
\begin{equation}
    i_d(k)+\sum_{n=1}^N i_n(k) =\sum_{m=1}^M i_m^\text{PW}(k) \doteq i_\text{total}(k). \label{cPWLcoupl}
\end{equation}

Also, we enforce limits on the variable associated with each linear segments:
\begin{equation}
        i_m^\text{PW}(k)\in[0,\Delta i].\label{cPWLipw}
\end{equation}

The PWL formulation adds one more set of box constraint than the NLP formulation and replaces the optimization variables $\mathbf i_\text{total}  \in \mathbb{R}^K$ with $\mathbf i^{PW} \in \mathbb{R}^{MK}$. This open-loop optimal control problem is then implemented in receding-horizon fashion as illustrated in Fig.~\ref{fig:MPC}.

\begin{remark}[Extending Theorem~\ref{thm1} to the PWL formulation]
Since the PWL relaxation overestimates the non-convex equality constraint, it is contained within the SOCP relaxation (as in Fig.~\ref{fig:Relax}). This ensures, under the same conditions of Theorem~\ref{thm1} and Corollary~\ref{corol1}, that the optimal solution from the PWL formulation is tight relative to the PWL segments. Thus, the PWL formulation can successfully predict and regulate the transformer's dynamic temperature trajectory relative to its temperature limit. For a detailed treatment of the PWL relaxation, please see~\cite{madsPWLMPC}. 
\end{remark}

\begin{figure}
    \centering
   \resizebox{\columnwidth}{!}{
      \begin{tikzpicture}[>={Stealth[inset=1pt,length=8pt,angle'=40,round]}]
            \node[dashed,rectangle, draw=black, thick,minimum height=10em, minimum width=14em, rounded corners] (MPC) at (0,0){};
            \node[rectangle, draw=black,minimum height=7em,minimum width=12em, rounded corners,label=above:{Open Loop Optimization}] (Opt) at (0,0){};
            \node[rectangle, draw=black,minimum height=4em,minimum width=12em, rounded corners,label=above:{QP Optimization Solver}] (G) at (0,0){};
            \node[rectangle, draw=black] (PWL) at (0,0.35){PWL Transformer Model};
            \node[rectangle, draw=black] (evmodel) at (0,-0.35){EV Charging Model};
            \node[rectangle, draw=black] (NLplant) at (7,0.25){Nonlinear Transformer};
            \node[rectangle, draw=black, below=0.5cm of NLplant.west, anchor=west, minimum width=10em](evmodel){EV Batteries};
            \node[font=\fontsize{15pt}{15pt}\selectfont] (Ta) at (-3,-1){$T_a$};
            \node[font=\fontsize{15pt}{15pt}\selectfont] (Id) at (-3,1){$i_d$};
            \draw[->,thick] (Opt.east)  -- node[above, font=\fontsize{15pt}{15pt}\selectfont] {$\{i_n(0)\}_{n=1}^N$}(5.2,0);
            \draw[->,thick] (Ta.east)  -- (Opt);  
            \draw[->,thick] (Id.east)  -- (Opt);  
            \draw[->,thick] (evmodel.south)  -- (7,-2.5) --node[above, font=\fontsize{15pt}{15pt}\selectfont] {$T_{\text{meas}}$, $\{s_{\text{meas},n}\}_{n=1}^N$} (0,-2.5) --(Opt.south);  
        \end{tikzpicture}
        }
    \caption{The OCP with feedback. The OCP is used with the ALADIN, ADMM, and dual decomposition methods and employs the PWL approximation of the transformer's nonlinear current-temperature relations in the OCP formulation while the plant model represents the non-linear transformer.}
    \label{fig:MPC}
\end{figure}
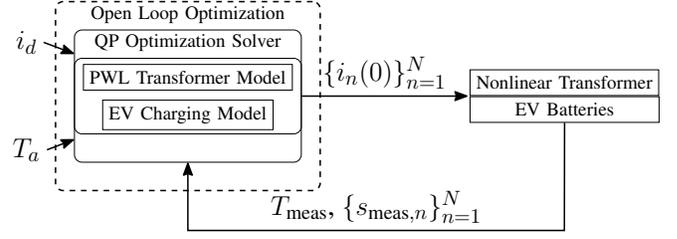

\section{Non-centralized Implementation} \label{sec:distributed}
The centralized problem can be decomposed into $N$ sub-problems if it were not for the \textit{coupling constraints}~\eqref{cPWLcoupl}. Thus, in this section, we present different distributed and decentralized charging algorithms. Specifically, ALADIN and PEM represent two novel contributions for EV charging control while the other two methods (Dual Ascent and ADMM) serve as base cases for comparison. 

The iterative ALADIN, Dual Ascent, and ADMM schemes employ the partial Lagrangian with respect to~\eqref{cPWLcoupl} as follows:

\begin{align}
    &\mathcal{L}(\mathbf{i}_n,\mathbf{s}_n,\mathbf i_m^\text{PW},\mathbf{\lambda})= \notag\\
    &\phantom{=}\sum_{n=1}^N J_n(\mathbf{i}_n,\mathbf{s}_n)+\mathbf{\lambda}^\top\left(\mathbf i_d+\sum_{n=1}^N \mathbf{i}_n  -\sum_{m=1}^M \mathbf i_m^\text{PW} \right) \nonumber \\
    &=\sum_{n=1}^N\left( J_n(\mathbf{i}_n,\mathbf{s}_n)+\mathbf{\lambda}^\top \mathbf{i}_n\right)+\mathbf{\lambda}^\top(\mathbf i_d-\sum_{m=1}^M \mathbf i_m^\text{PW}), \label{sepL}
\end{align}
where $\mathbf{\lambda}\in \mathbb{R}^{K}$ are the Lagrange multipliers associated with~\eqref{cPWLcoupl}. From~\eqref{sepL}, the Lagrangian can be separated into local EV variables $\{\mathbf{i}_n,\mathbf{s}_n \} \in \mathbb{R}^{2NK}$ and transformer variables $\{\mathbf i^{PW}\}\in \mathbb{R}^{MK}$, which turns~\eqref{sepL} into a separable objective function subject to decoupled constraints. 
This means that the optimization problem can be solved in distributed fashion by iteratively updating $\mathbf{\lambda}$ for which we develop and present 
Dual Ascent, ADMM and ALADIN algorithms. We also provide a non-iterative packet-based coordination scheme adapted from PEM. Each algorithm has different requirements for the transformer, EVs, and coordinator problems as illustrated in Fig.~\ref{fig:distG}. 
Next, we will discuss each scheme and since the dual decomposition and ADMM are two common methods, details have been omitted in this manuscript.

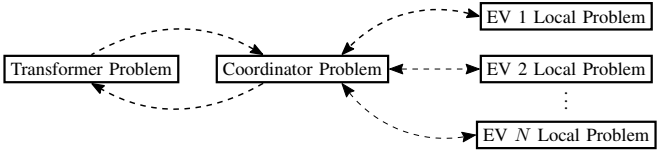
\begin{figure}[t]
    \centering
    \resizebox{\columnwidth}{!}{
        \begin{tikzpicture}[>={Stealth[inset=1pt,length=8pt,angle'=40,round]},
            squarednode/.style={rectangle, draw=black, very thick, minimum size=5mm},
            ]
            \node[squarednode] (C) at (0,0){Coordinator Problem};
            \node[squarednode] (EV1) at (5,1){EV 1 Local Problem};
            \node[squarednode] (EV2) at (5,0){EV 2 Local Problem};
            \node[] (EV3) at (5,-0.5){$\vdots$};
            \node[squarednode] (EV3) at (5,-1.25){EV $N$ Local Problem};
            \node[squarednode] (XFRM) at (-4,0) {Transformer Problem};
            
            \draw[dashed,->,thick] (C.200) to [bend left]  (XFRM.south);
            \draw[dashed,<-,thick] (C.160) to [bend right]  (XFRM.north);
            \draw[dashed,<->,thick] (EV1.west) to [bend right] (C.20);
            \draw[dashed,<->] (C.east) to  (EV2.west);
            \draw[dashed,<->] (C.340) to [bend right] (EV3.west);
        \end{tikzpicture}
    }
    \caption{Distributed EVC coordination scheme. Synchronous schemes require the communications from all local problems before proceeding with the coordinator step.}
    \label{fig:distG}
\end{figure}

\subsection{ALADIN}
ALADIN  is a relatively new distributed optimization algorithm~\cite{ALADIN}. It has been considered, among other things, for optimal power flow problems~\cite{kit:engelmann18b}.
The method decomposes the centralized optimization problem by having each agent solve its local problem based on the primal iterate guess of primal variables and the Lagrange multipliers of the coupling constraints. 
The local solutions together with first- and second-order information are provided to the coordinator to update the primal variables and multipliers by solving a centralized (but simple) quadratic program (QP) to foster consensus. 
This setup allows predictive ALADIN to achieve quadratic convergence locally, which greatly reduces the number of iterations needed and is highly desirable.


The first two steps of ALADIN are shown in Algorithm~\ref{ALADalg} and solve local optimization problems for: 1) the EVs and 2) the transformer. Here, the primal variables for the EVs are collected in $ \mathbf x_n^\top  = \left ((\mathbf i_n)^\top , \; (\mathbf s_n)^\top \right )$ and the primal variables for the tranformer in 
$\Delta \mathbf x_{N+1}^\top  = \left ((\mathbf i_{N+1})^\top , \; (\mathbf T)^\top\right  )$ with $(\mathbf i_{N+1})^\top  = (\mathbf i^{PW\top}_1,\dots,\mathbf i^{PW\top }_M)^\top $, where the latter
is due to the PWL formulation~\eqref{eq:PWL}.
As mentioned, the primal solution from each local EV and transformer problem is shared with the coordinator. 
In addition, the gradient and Hessian of the Lagrangian relative to $\mathbf x_n$ is denoted $g_{x_n}$ and $H_{x_n}$. 
Note that the gradients of the box constraints $C_{\bar{x}_n}^{(p)}$ are constant and given by zero vectors with $-1$ or $1$ in the column corresponding to a primal variable for which a box constraint is active.
So here it suffices to communicate an index set of the active constraints instead of a full matrix. 
In Step 3), the coordinator combines the local information into a coordination QP to update the auxiliary variables and the dual variables. 
The specific ALADIN variant used for the EV charging OCP 
is shown in Algorithm~\ref{ALADalg}. 
A slight alteration to the ALADIN formulation is used which changes the linearized expressions, $C_{\bar{x}_n}^{(p)} \Delta\mathbf{x}_n$, in Step 3) to be inequalities from their original equality constraints. 
This relaxation allows the local variables to move asymmetrically away from its bound instead of fixing all variables that are at their upper or lower limit.

\begin{algorithm}[t]
    \caption{ALADIN for EV charging.}
    \label{ALADalg}
    \begin{algorithmic}[1]
        \Statex \textbf{Initialization}: Initial ($p\equiv 0$) guess of dual multiplier $\mathbf{\lambda}^{(0)}$, of all four auxiliary variables $\{{\mathbf{i}}_{n}^{(0)},{\mathbf{s}}_{n}^{(0)},\mathbf{i}_{N+1}^{PW\, (0)},\mathbf{T}^{(0)}\}$ and tuning parameters $\{\rho_{\text{ALAD}},\mu, \sigma_z,\sigma_t,\left \{\sigma_{i,n},\sigma_{s,n} \right \}_{ n \in 1,\dots,N}\}$ .
        \Statex \textbf{Repeat for} $p$:
        \begin{enumerate}
            \item \textit{Solve local EV problems:} for each $n \in 1,\dots,N$ 
                \begin{align*}
                    \bar {\mathbf{i}}&^{(p)}_n=\arg \min_{\mathbf{i}_n,\mathbf{s}_n}\;\;J_n(\mathbf{i}_n,\mathbf{s}_n)+(\lambda^{(p)})^\top\mathbf{i}_n+ \\&\frac{\rho_{\text{ALAD}}\sigma_{i,n}}{2}(\mathbf{i}_n-  {\mathbf{i}}_{n}^{(p)})^2+\frac{\rho_{\text{ALAD}}\sigma_{s,n}}{2}(\mathbf{s}_n- {\mathbf{s}}_{n}^{(p)})^2  \notag \\
                    &\text{s.t.}\;\;\; \eqref{cNLsn}, \eqref{cNLsnmin}, \eqref{cNLin}	 \notag
                \end{align*}
            \item \textit{Solve local transformer problem:}
                \begin{align*}
                    (\bar {\mathbf{i}}&_{N+1})^{(p)}=\arg \min_{\mathbf{i}_{N+1},\mathbf T}\;\;-(\lambda^{(p)})^\top \sum_{m=1}^M \mathbf i^{PW}_{m}+\\&\frac{\rho_{\text{ALAD}}\sigma_{Z}}{2}(\mathbf{i}_{N+1}-\mathbf{i}_{N+1}^{(p)})^2\notag+\frac{\rho_{\text{ALAD}}\sigma_T}{2}(\mathbf{T}-\mathbf{T}^{(p)})^2 \\
                    &\text{s.t.}\;\;\; \eqref{cPWLt}, \eqref{cPWLipw}, \eqref{cPWLtmax} \notag
                \end{align*}
            \item \textit{Solve coordinator problem:}
                {\scriptsize
                \begin{subequations}
                    \begin{align*}
                         \min_{\Delta x_n, y}   \: &\sum_{n=1}^{N+1} \left( \frac{1}{2}\Delta\mathbf{x}_n H_{x_n}^{(p)})  \Delta\mathbf{x}_n +g_{x_n}^{(p)}\Delta\mathbf{x_n} \right)+(\lambda^{(p)})^\top y + \frac{\mu}{2}||y||_2^2\\
                        \text{s.t.} \: &\sum_{n=1}^{N} (\bar i_n^{(p)}(k)+\Delta \bar i_n(k)) -\sum_{m=1}^M \left((\bar i_m^\text{PW})^{(p)}(k)+\Delta i_m^\text{PW}(k)\right)=\notag\\&\quad y(k)-i_d(k) \quad |\: \lambda_{QP}(k) \\
                        &\Delta T(k+1)=\tau \Delta T(k)+\gamma \left( \Delta i \sum_{m=1}^M(2m-1)\Delta i_m^\text{PW}(k) \right) \\ 
                        &\Delta s_n(k+1)=\Delta s_n(k)+\eta_n \Delta i_n(k) \text{ for all } k=0,\hdots, K-1 \notag \\
                        &C_{\bar{x}_n}^{(p)} \Delta\mathbf{x}_n \leq 0 \label{jacobC1} \; \text{ for all }\; n=1,\hdots, N+1 \notag
                    \end{align*}
                    \end{subequations}
                    
                    }%
            \item \textit{Termination Criterion:} 
            If
            {\begin{align*}
                \left|\left| \bar i_d(k)+\sum_{n=1}^N \bar i_n(k)^{(p)}-\sum_{m=1}^M (\bar i_m^\text{PW})^{(p)}(k)\right|\right|_1 \leq \epsilon_1  \\ \text{ and }  
                \left|\left|\sigma_{x_n}(\bar {\mathbf{x}_n}^{(p)} -\mathbf{x}_n^{(p)})\right|\right|_1\leq \epsilon_2
            \end{align*}
            then exit with $x^*=x^{(p)}$ and $i_n^\ast(0)$ is implemented in EVs.}
            \item \textit{Update dual variable and auxiliary variables}
                \begin{align*}
                    \mathbf{x}_n^{(p+1)}&=\bar {\mathbf{x}}_n^{(p)}+\Delta \mathbf{x}_n, \quad \ n=1,\hdots, N+1 \\
                     \mathbf \lambda^{(p+1)}&=\mathbf \lambda_{QP}, \qquad 
                    p \rightarrow p+1. \notag
                \end{align*}
        \end{enumerate}
    \end{algorithmic}
\end{algorithm}
 
ALADIN provides a systematic approach to decomposing our large centralized primal formulation into many small, local QPs and a single coordination QP. 
However, despite the few iterations required for convergence (e.g., please see~\cite{ALADIN}), the information required from the sub-problems is significant and the coordinator problem is computationally intensive. For a variant of ALADIN with reduced size of the coordination QP we refer to~\cite{kit:engelmann19a}.
Note that the ALADIN tuning parameters  $\rho_{\text{ALAD}},\mu, \sigma_z,\sigma_t$ and $\sigma_{i,n},\sigma_{s,n} ~\forall n$ have to be chosen initially.

\subsection{Dual Decomposition}
Dual decomposition separates~\eqref{sepL} and creates local QP EV problems, a local LP transformer problem, and updates $\mathbf{\lambda}$ by dual ascent. Standard dual decomposition with dual ascent update for separable problems is used from~\cite{boydDecomp} and the setup is similar to that of~\cite{EVC}, except herein we employ the relaxed PWL model and not a linearized model. 

\subsection{ADMM}
ADMM builds on top of dual decomposition by augmenting the local objective functions using auxiliary variables. A separable ADMM approach is used herein based on~\cite{bertsekas_2015}.

Note that both dual decomposition and ADMM can be expressed as a special case of ALADIN with considerably simplified coordination QP, and, in case of dual decomposition, by choosing all $\sigma_{(.)}$'s to zero additionally (cf.~\cite{ALADIN}). Also, these classical optimization methods  are only used as bases for comparison in the case studies. For these reasons and due to the fact that detailed treatments of these algorithms for EV charging problems are widely available in the literature, we do not state them explicitly here.  

\subsection{PEM with Dynamic Constraints}
PEM represents a computationally and informationally light demand-side coordinating scheme for coordinating DERs (in real-time), such as EV chargers. The scheme uses a stochastic, packet-based approach similar to modern communication networks to dynamically prioritize demand-side resources based on local energy needs~\cite{PEM}. The full PEM algorithm adapted for the EV charging problem approximates the OCP and is described in Algorithm~\ref{PEMalg}. Each local ``packetized'' charger can infer or measure its local \textit{energy need}, which is mapped to a prescribed probability of requesting a fixed-duration ($\delta>0$) packet of energy (e.g., a $\delta=5$ step, constant-ampere charging epoch). The request is submitted to the coordinator,  which takes into account real-time and/or predicted transformer conditions to either accept or reject the packet to maintain the transformer temperature within its limits. To ensure quality-of-service (QoS) for the device owner, opt-out logic enables EVs with immediate energy needs to temporarily exit the scheme and recover their SoC. Algorithm~\ref{PEMalg} is described next.

\subsubsection{Local EV Problem}
The PEM scheme does not require that the EVC agent solves a local optimization problem to manage EV charger demand. Instead, a ``packetized'' EV charger is assumed capable of accurately measuring the EV's SoC, $s_{\text{meas},n}$, and inferring time until departure, $\bar{k}_n$. Based on these two updates, the EV charger calculates its energy need with the ratio 
$$\text{ratio}_n(k) \doteq \frac{\bar s_n -s_{\text{meas},n}(k)}{\eta_n i_n^\text{max} (\bar k_n - k)} \in \mathbb{R}.$$
If the ratio$_n$ $>1$, then the time remaining is not sufficient to provide the desired energy, even if charging for the entire remaining duration. Thus, if ratio$_n$ reaches or initially exceeds unity, then the device will automatically opt out (opt out status denoted by $\text{Req}_n < 0$) and continuously charge until the time of departure. Thus, opting out represents a background disturbance to the fleet of packetized EV chargers, which reduces the number of packets that can be accepted by the coordinator. When $\text{ratio}_n \in [0,1]$, the value is mapped to a probability of requesting a packet over the duration of time step $k$ (request status denoted by $\text{Req}_n \in \{0,1\}$), where a request from EV $n$ is sent to the coordinator with exponential waiting times, e.g, please see Algorithm~\ref{PEMalg}. The probability of requesting a packet depends on the ratio and a pre-specified mean time-to-request (or \textit{mttr}) for a specific ratio value set-point ($\hat r_{\text{set},n} \in (0,1)$). As ratio$_n(k)\rightarrow 0/1$, the probability of requesting a packet during timestep $k$ approaches~$0/1$. Of course, while the charger is ``consuming'' an energy packet, it does not request another packet, so Req$_n(k)$ status is set to the negative of the packet completion timer.

If an EV requests a packet ($\text{Req}_n \equiv 1$) and is notified that its packet is accepted, the EV charges at a pre-specified current for $\delta$ time steps. If the ratio$_n<0$, we denote status by $\text{Req}_n \equiv 2$, which implies that the EV's SoC exceeds its desired (minimum) energy target, which means that the EV's local ``energy need'' has been satisfied and any future requests from this EV is designated a low-priority request.

\subsubsection{Coordinator Problem}
The coordinator receives packet requests and must now accept a proportion of them in such a way as to keep the transformer temperature within limits. Since temperature is a dynamic state and prior work with PEM and EVs focused on static power or current limits one major contribution of this paper is the extension of PEM for scheduling under dynamic state constraints. Thus, this section extends prior work on PEM with a novel, predictive, synchronous coordinator formulation that utilizes an efficient MIQCQP formulation to select which requests are accepted and denied. Note that the predictive model in the coordinator only concerns the transformer temperature relative to changes in demand and does not extend to the fleet's requests or SoC. That is, the coordinator uses a persistent forecast of requests and opt-outs over its prediction horizon, which is just $\delta$ timesteps (e.g., $\delta=2$ timesteps, which is 6~minutes in  Case-study~1). 

To do so, first define the set of all devices that do not request a packet (\text{Req}$_n(k)\equiv 0$) during timestep $k$ as $\mathcal{E}_0$. The EVs that request a packet at time $k$ belong to set $\mathcal{E}_1$. Then, define $\delta$ sets for the devices that are ``locked in'' for future time steps as $\mathcal{Y}_{l}, l\in {k+1, \hdots, k+\delta}$ to capture the groups of EVs still consuming an energy packet or those that opted out earlier. Lastly, define $\mathcal{E}_2$ as the set of EVs that have already reached their desired (minimum) SoC target, but are not fully charged. 

The coordinator's MIQCQP problem is solved in Algorithm~\ref{PEMalg} to determine which EVs have their packet requests accepted ($\text{Resp}_n(k)=1$) and rejected ($\text{Resp}_n(k)=0$). Since the problem looks ahead just a packet length, the prediction horizon is short and the formulation is efficient. The requests from $\mathcal{E}_2$ are de-prioritized by use of a scaling factor ($\omega_E \doteq \min\{1/(NK), 1/(4N)\} << 1$) in the objective function. To ensure a solution always exists, a slack variable is added to the temperature limit and penalized in the objective function ($\omega_S>>1$). 

Finally, 
the MIQCQP depends on the EV chargers' ampacities, $i_n^\text{max}$. The current rating is known exactly when the information is included in the request or may be approximated via data-driven methods. In this work, we assume the former. After solving the MIQCQP, the optimal solution, $\mathbf{u}^\ast_\text{ch}(0) \in \mathbb{R}^N$, represents the EV chargers whose requests were accepted by the coordinator. To reduce the necessary communications in an online implementation of PEM, only EVs whose charger's logic state undergoes a transition (e.g., $\text{Resp}_n(k) \neq \text{Resp}_n(k-1)$) 
are updated by the coordinator. 


The next two sections explore the convex transformer temperature models and EV charging algorithms within two different, but relevant scenarios: residential and commercial fleets of EVs. The latter will require novel modeling of a hub of commercial EVs and then adapt the algorithms to manage the charging of multiple hubs rather than individual EVs. Finally, Section~\ref{sec:comp} provides discussion and comparison between distributed methods employed in case studies~1 and~2.

\algblock[Name]{Local}{End}
\algblock[Name]{Coordinator}{End}
\algblock[Name]{Transformer}{End}

\begin{algorithm}[ht]
    \caption{Look-ahead PEM algorithm}
    \label{PEMalg}
    \begin{algorithmic}
        \Local \textbf{ EV Problem:} compute for each $n\in N$
                \If {Consuming Packet} 
                \State $\text{Req}_n(k)=- \text{ duration remaining for packet}$
            \ElsIf {$s_{\text{meas},n}(k)=1$}  \Comment{EV at 100\%}
                \State $\text{Req}_n(k)=0$ 
            \ElsIf {$ s_{\text{meas},n}(k)\geq \bar s_n$} \Comment{EV is low priority}
                \State $\text{Req}_n(k)=2$
            \Else
                    \State $\text{ratio}_n(k)=  \dfrac{\bar s_n-s_{\text{meas},n}(k)}{\eta_n* i_n^\text{max}(\bar k_n-t)}$ 
                    \If {$\text{ratio}_n(k) \ge 1$}\State{$\text{Req}_n(k)=-\delta$} \Comment{EV opts out}
                    \Else
                          \State $\mu(k)=\frac{1}{mttr}\frac{\text{ratio}_n(k)}{1-\text{ratio}_n(k)}\frac{1- r_{\text{set},n}}{ r_{\text{set},n}}$
                        \State $P_n(k)=\min\{\max \{1-e^{-\mu(k)\Delta t},0\},1\}$
                        \State $ \text{Req}_n(k)= \begin{cases}1, & rand() < P_n(k)\\0, &else \end{cases}$
                    \EndIf
            \EndIf
        \End
        \Coordinator \textbf{ Problem:}\\
            {\small
             Update sets $\mathcal{E}_0, \mathcal{E}_1, \mathcal{E}_2, \mathcal{Y}_k$ and measure $T_{\text{meas}}$\\
            \text{Solve MIQCQP to determine $\mathbf{u}_\text{ch}^\ast(k) \,\, \forall k\in \mathcal{K}_\delta \doteq \{0,\hdots,\delta-1\}$}
                \begin{align}
                    \max_{T_\text{slack},u_n} \: & \sum_{k=0}^{\delta-1} \left( \sum_{n\in {\mathcal{E}}_1} u_n(k) +\omega_E\sum_{n\in\mathcal{E}_2} u_n(k)\right)  - \omega_S T_\text{slack}\\
                    \text{s.t } \: 
                    & T(k+1) = \tau T(k) +\gamma e(k) +  \rho T_a(k) \,\, \forall k \in \mathcal{K}_\delta\\
                    &e(k) \ge \left(i_\text{total}(k)\right)^2 \quad \forall k \in \mathcal{K}_\delta\\
                    & i_\text{total}(k) = i_d(k)+\sum_{n=1}^N u_n(k) i_n^\text{max} \quad \forall k \in \mathcal{K}_\delta\\
                    &T(k+1)\leq T^\text{max} +T_\text{slack} \quad \forall k \in \mathcal{K}_\delta\\
                    &u_n(k)=1 \quad \forall n \in \mathcal{Y}_{k}\,\, \forall k \in \mathcal{K}_\delta\\
                    &u_n(k)\le u_n(k+1) \quad \forall n \in \mathcal{E}_1 \cup \mathcal{E}_2 \quad \forall k \in \mathcal{K}_\delta\\
                    & T(0)=T_{\text{meas}} \\
                    &\sum_{n\in\mathcal{E}_0} \sum_{m=0}^{\delta-1}  u_n(m)=0 \\
                    &u_n(k) \in \{0,1\} \quad \forall n=1,\hdots,N \: \forall k \in \mathcal{K}_\delta
                \end{align}
            }%
            {\small Determine responses to packet requests from EV chargers:} 
            \State $\text{Resp}_n(k)= u_{\text{ch},n}^*(0)\,\, \forall n$ \Comment{From optimal solution.}
        \End
    \end{algorithmic}
\end{algorithm}

\section{Case Study 1: Residential PEV charging} \label{sec:evc}
Now that we have developed several distributed control methods for the EV charging problem, we would like to simulate a scenario and evaluate each method on privacy, performance, and processing metrics. We simulated a scenario with 100 EVs, for the overnight hours of 8PM to 10AM. The rest of the parameters used are shown in Table \ref{table:simP} where the bracket notation $[a, b]$, denotes the range of the variable. For the look-ahead PEM, we use $\delta = 2$ timesteps,  $mttr=\Delta t\delta=2$ timesteps (i.e., $2\times 180=360$s), and $\hat r_{\text{set},n}=0.10$.

\begin{table}
    \begin{center}
    \caption{Simulation parameters for case study~1}
    \scalebox{0.9}{
         \begin{tabular}{l c c c}
         \toprule
         \textbf{Variable} & \textbf{Value} & \textbf{Units} & \textbf{Source}\\
      \midrule
         \multicolumn{4}{c}{\textit{System parameters}} \\
         \midrule
         $N$ & 100 & EVs & \cite{umassEV,evPenetration} \\
         $K$ & 160 & Timesteps & -\\
         $M$ & 6 & Segments & - \\
         $\Delta t$ & 180 & Seconds & \cite{xfrmStep,hotspotXFRM}\\
         $ \tilde T_a(k)$ & $[16, 18]$ & $^\circ$C &-\\
         $ c$ & 29.87 & & \cite{paulXFRM,xfrmStep,hotspotXFRM} \\
         $i_d(k)$ & $[12.1,17.5]$ & kA&\\
         $T^\text{max}$ & 100 & $^\circ$C &  \cite{xfrmGuide,united1991permissible} \\
         $T_0$ & 70 & $^\circ$C& - \\
         $\gamma$ & 0.0131 &  $^{^\circ C}/_{(\text{kA})^2}$ & \cite{paulXFRM,xfrmStep,hotspotXFRM} \\
          $\tau$ & 0.9145 & & \cite{paulXFRM,xfrmStep,hotspotXFRM}\\
          $\rho$ & 0.0855 & & \cite{paulXFRM,xfrmStep,hotspotXFRM}\\
        $R$ & 240/8320 & & -\\
          \midrule
         \multicolumn{4}{c}{\textit{Residential EV parameters}} \\
         \midrule
         $s_{n,0}$ & $[0, 70]$ & \% & \cite{fleetcarma}\\
         $i_n^\text{max}$ & $[12,80]$ & A & \cite{umassEV, teslaWall}\\
         $\alpha_n$ & $[80, 90]$ & \% & \cite{veicEff} \\
         $\beta_n$ & $[40,100]$ & kWh & \cite{evSales,umassEV}\\
         $\bar s_n$ & $[75, 100]$  & \%&-\\
         $\bar k_n$ & [06:00,10:00] & Time &-\\
         $ q_n$, $r_n$ & $[0,50]$, $10$ & &-\\
         \bottomrule
        \end{tabular}
    }
    \label{table:simP}
    \end{center}
\end{table}

\begin{figure}
    \includegraphics[width=\columnwidth]{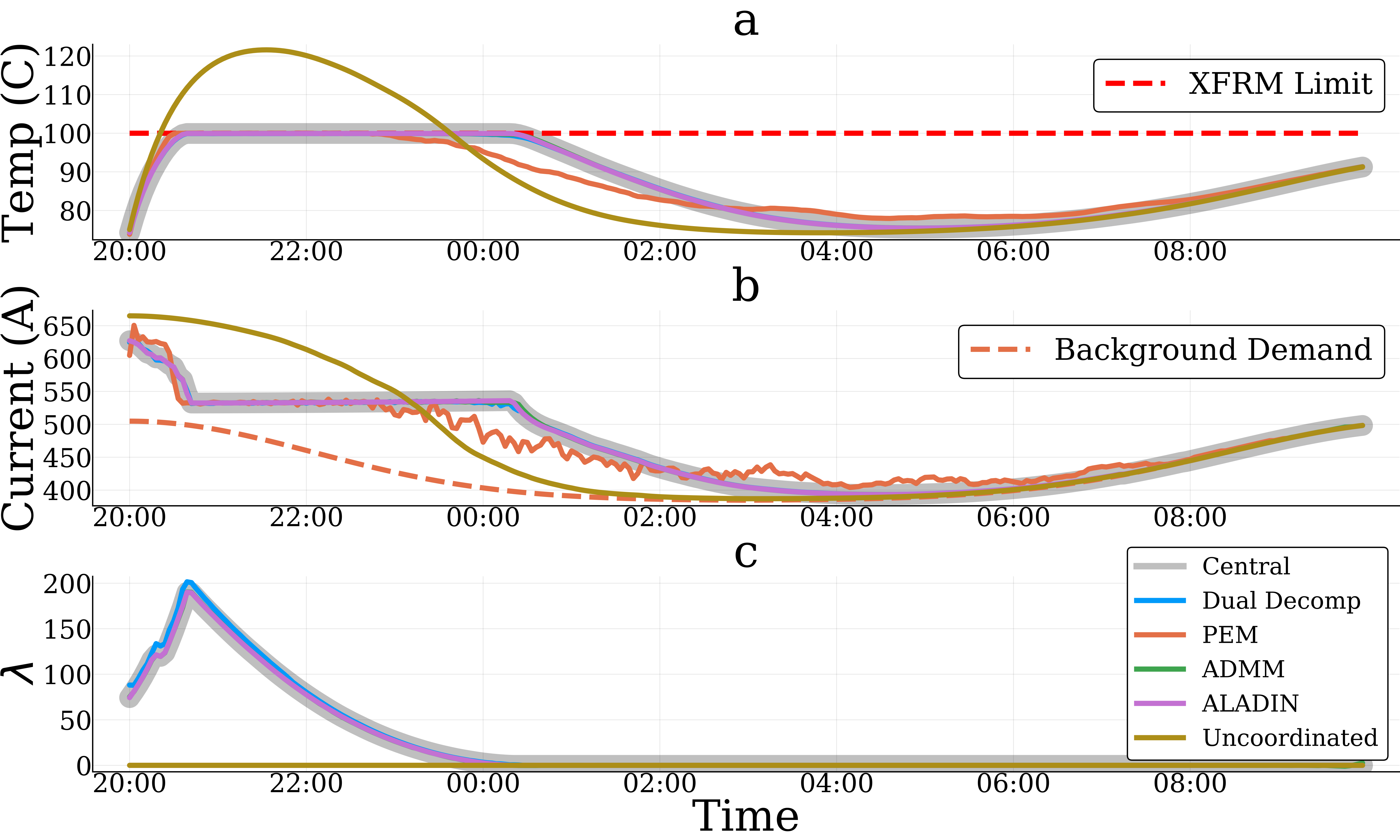}
    \caption{Case Study~1 receding-horizon response. \textit{Top}: Temperature response; \textit{Middle}: Total primary network current demanded from substation transformer and equal to reflected total secondary current ($R i_\text{total}(k)$); \textit{Bottom}: Dual variable of~\eqref{cPWLcoupl}. }
    \label{case1Results}
\end{figure}

\begin{figure}
    \centering
    \includegraphics[width=\columnwidth]{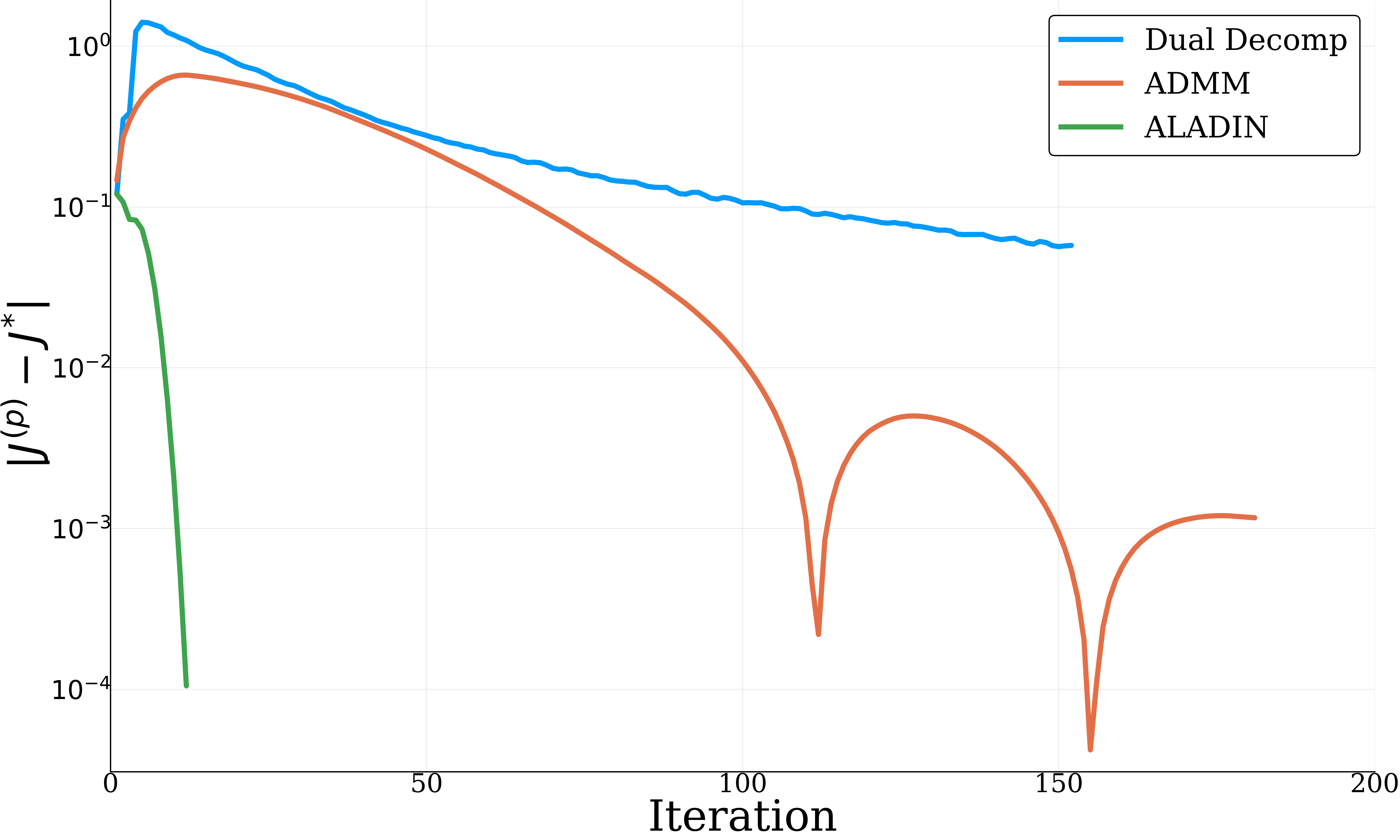}
    \caption{Case Study~1 convergence for first (cold-start) time instance. Objective function values converging to centralized (optimal) value.}
    \label{case1Conv}
\end{figure}



\subsection{Simulation Results}
The OCP is solved in closed-loop to engender an receding-horizon simulation for each non-centralized algorithm, as well as, the centralized formulation is shown in Fig.~\ref{case1Results}. The top two plots show the transformer temperature, (a), and total primary network load at the transformer, (b). The bottom plot, (c), displays the dual multiplier $\lambda$, which is associated with the coupling constraint~\eqref{cPWLcoupl}. In addition to the centralized OCP solution and the solutions from the four non-centralized algorithms, the result of the uncoordinated EV charging is also provided. Clearly, without EV charging control implemented, the transformer temperature exceeds its limit for hours.

The optimal solutions of ADMM and ALADIN are nearly identical to the centralized solution. The convergence of the three iterative schemes in solving the OCP for the first (cold-start) timestep of the receding-horizon simulation is shown in Fig.~\ref{case1Conv}. ALADIN significantly outperforms ADMM, which outperforms the Dual Decomposition. 

For the scenario in Case Study~1, privacy is important as residential EV owners are not inclined to share information about their driving habits, such as time of arrival, departure, and state-of-charge. Thus, since the ALADIN algorithm requires significant information transfer  between EVs and coordinator, ALADIN is not well-suited for residential charging applications. Due to its quadratic convergence, ALADIN may be an ideal approach for solving EV charger problems where privacy is less important and full information is available. A fleet of commercial EVs located at central charging hubs within a large city fits those conditions and is presented next.

\begin{figure}[ht]
    \centering
    \scalebox{.8}{
        \begin{tikzpicture}[>={Stealth[inset=1pt,length=8pt,angle'=40,round]},
        squarednode/.style={rectangle, draw=black, very thick, minimum size=5mm},
        ]
        \node[squarednode] (XFRM) at (0,0){Transformer};
        \node [squarednode] (hub1) at (2,-1){Hub 1};
        \node [squarednode] (hub2) at (2,1){Hub 2};
        \node [squarednode] (hub3) at (4,-1){Hub 3};
        \node [squarednode] (hub4) at (4,1){Hub 4};
        \node [] (dots) at (5,0){$\hdots$};
        \node [squarednode] (hubh) at (7,0){Hub $H$};
        
        \draw[->,thick] (XFRM.east) -| node[left,near end] {$i_1$} (hub1.north);
        \draw[->,thick] (XFRM.east) -| node[left,near end] {$i_2$} (hub2.south);
        \draw[->,thick] (XFRM.east) -| node[left,near end] {$i_3$} (hub3.north);
        \draw[->,thick] (XFRM.east) -| node[left,near end] {$i_4$} (hub4.south);
        \draw[-,thick] (XFRM.east) -- (dots.west);
        \draw[->,thick] (dots.east) -- node[above] {$i_H$} (hubh.west);
        \end{tikzpicture}
    }
    \caption{Case Study~2 network of commercial EV charging hubs, where the  transformer's low-voltage (LV) side is rated at $V_\text{pri}=13.2$kV while the commercial chargers are supplied at $V_\text{sec}=480$V in the secondary network.}
    \label{case2model}
\end{figure}
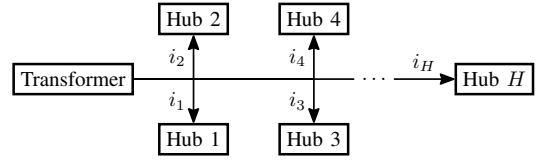

\section{Case Study 2: Commercial fleet charging} \label{sec:hub}

The residential EV optimal charging problem presented in Section~\ref{sec:evc} is just one scenario where coordinating the demand related to EV charging will be useful. As commercial transportation becomes electrified, these vehicle fleets will also benefit from scheduled charging control. In addition, a large proportion of these fleets have predictable routes to and from a central depot. These central depots or ``hubs'' represents local EV charging centers.

The commercial/industrial EV hub charging problem is inherently different from the residential charging problem. For example, the privacy of an individual vehicle in a commercial fleet is not a concern as one company owns and centrally dispatches all EVs in their fleet. The vehicles will have larger batteries and since there are multiple electric vehicles in the hub, each hub represents a much larger total demand on the system than an individual EV from Case Study~1. 

In the formulation below, we represent each hub as a single agent in the system and assume that internal to each hub is an algorithm that distributes allocated hub charging capacity to its individual vehicles. Thus, the hub node needs to aggregate available EV SoC and energy limits to ensure that the hub can meet the underlying, asynchronous EV charging needs. Next, we develop the dynamic model of a single hub and present the distributed optimal charging control problem for a collection of $H$ hubs under a large MVA-scale transformer.
 
\subsection{Hub System Model}

Define hub $h\in\{1,\hdots, H\}$ by a set of $N_h$ assigned vehicles $\mathcal{N}_h$. Each vehicle $n \in \mathcal{N}_h$ has battery capacity $E_{h,n}^{\text{max}}$, is predicted to arrive at time step $\underline{K}_{h,n}$ with arrival energy $\underline s_{h,n} E_{h,n}^{\text{max}}$, and is predicted to leave at time step $\bar K _{h,n}$ with desired minimum departure energy $\bar s_{h,n} E_{h,n}^{\text{max}}$. The value $\underline{s}_{h,n}$~($\bar s_{h,n})$ represents the vehicle's relative SoC at arrival (departure). For each hub and each time step, we then define dynamic sets for arriving, parked, and departing vehicles:
\begin{align}
    \text{Arrive}_{h}(k)=&\{n \in \mathcal{N}_h|  k = \underline{K}_{h,n}\}\\
    \text{Parked}_{h}(k)=&\{n \in \mathcal{N}_h| \underline{K}_{h,n} < k < \bar{K}_{h,n}\}\\
    \text{Depart}_{h}(k)=&\{n \in \mathcal{N}_h|  k = \bar{K}_{h,n}\}
\end{align}
and calculate the arrival and departure energy trajectories 
\begin{align}
    E_{h,\text{arrive}}(k)&=\sum_{n\in\text{Arrive}_{h}(k)}\underline s_{h,n}E_{h,n}^{\text{max}}\\
    E_{h,\text{depart}}(k)&=\sum_{n\in\text{Depart}_{h}(k)}\bar s_{h,n} E_{h,n}^{\text{max}}. 
\end{align}
These trajectories define the amount of energy added and subtracted from the predicted vehicle arrivals and departures, respectively. From the parked vehicles in hub $h$ at time $k$, we can also define the time-varying upper limits on energy
and effective current:
\begin{align}
   E_h^{\text{max}}(k)&=\sum_{n\in\text{Parked}_{h}(k)} E_{h,n}^{\text{max}}. \\
   i_h^{\text{max}}(k)&=\sum_{n\in\text{Parked}_{h}(k)} i_{h,n}^{\text{max}}.
\end{align}
Note that although the maximum current capacity of the charging facility would not change physically, the effective maximum current at time $k$ is a function of the number of parked vehicles:

 Finally, since vehicle $n$ can depart from a hub $h$ with more than its desired departure SoC $\bar s_{h,n} E_{h,n}^{\text{max}}$, we need to account for the difference between the expected departing SoC and actually departing with up to 100\% of SoC:
\begin{align}
  E_{h,\Delta}^{\text{max}}(k)&=\sum_{n\in\text{Depart}_{h}(k)} (1-\bar s_{h,n}) E_{h,n}^{\text{max}}   
\end{align}
From these sets and trajectories, we can now form the hub energy dynamics and optimization.

\subsection{Hub Energy Dynamics and Bounds}
The aggregated SoC for each hub at time $k+1$ is a function of the current delivered over timestep $k$, the expected energy lost from departing vehicles, and the expected energy gained from arriving vehicles. The departed energy from each time step is the expected target SoC for the departing vehicles plus any extra energy provided, $E_{h,\text{depart}} + E_{h,\Delta}$, to bring (some) vehicles above required $\bar s_{h,n}$ to the full SoC. 
\begin{align}
    E_h(k+1)=&E_h(k)+\eta_h i_h(k)+E_{h,\text{arrive}}(k) -\notag\\&(E_{h,\text{depart}}(k)+E_{h,\Delta}(k))\\
    0\leq &E_h(k)\leq E_h^{max}(k) \\
    0\leq&E_{h,\Delta}(k)\leq    E_{h,\Delta}^{max}(k)\\
    0\leq &i_h(k) \leq  i_h^{\text{max}}(k)
\end{align}
\begin{remark}
Note that the hub charging efficiency parameter $\eta_h$ is assumed to be time-invariant (i.e., all vehicles charge with the same efficiency). However, $\eta_h$ could be estimated based on a weighted combination of the efficiencies in Parked$_h(k)$.  
\end{remark}
Next, we present the subtle differences between the hub system and residential charging problems as it relates to the objective function.

\subsection{Objective Function with Hubs}
The local hub objective function is similar to the one in the local (residential) EV scenario. However, now we want to minimize the deviation of the predicted hub energy level to its maximum possible energy state, which is the sum of the energy capacities for all vehicles forecasted to be parked at their respective hubs. In addition, if it is possible, it is desirable to maximize the $E_{h,\Delta}$ terms as they allow the hub to maximize the underlying EV SoC. The weighting factor $o_h$ determines how desirable oversupplying energy is relative to the weights of the other two terms ($q_h$ and $r_h$) and defines objective $h$ as
\begin{align}
   J_h(\mathbf{i}_h,\mathbf{E}_h,\mathbf{E}_{h,\Delta})=\sum_{k=1}^K q_h (E_h(k)- E_h^{\text{max}}(k))^2 + \\r_h (i_h(k))^2 - o_h E_{h,\Delta}(k) \notag
\end{align}
The centralized OCP for a system of hubs is presented next. 


\subsection{Centralized Optimal Control Problem with Hubs}
With the same PWL approximation of the transformer model as in~\eqref{cPWLt}, we can combine the hub dynamics and objective function to yield a centralized hub OCP:
\begin{subequations}
\begin{align}
\min\;\;&\sum_{h=1}^H J_h(\mathbf{i}_h,\mathbf{E}_h,\mathbf{E}_{h,\Delta}) \\
\text{s.t.}\;\;\;	& \notag\\
&E_h(k+1)=E_h(k)+\eta_hi_h(k)\notag +E_{h,\text{arrive}}(k)\\ &\qquad -(E_{h,\text{depart}}(k)+E_{h,\Delta}(k))\\
&T(k+1)=\tau T(k) +\rho T_a(k) \notag \\
& \qquad +\gamma \left(\Delta i \sum_{m=1}^M(2m-1)i_m^\text{PW}(k) \right) \\
&i_d(k)+\sum_{h=1}^H i_h(k)  = \sum_{m=1}^M i_m^\text{PW}(k)  \quad |\: \lambda(k) \\
&0\leq E_{h,\Delta}(k)\leq   E_{h,\Delta}^{\max}(k) \\
&0\leq i_m^\text{PW}(k)\leq\Delta i \quad \forall m=1,\hdots,M\\
&0\leq i_h(k)\leq i_h^{\text{max}}(k) \\
&0\leq E_h(k+1)\leq  E_h^{\max}(k), \,\, E_h(0) = E_{\text{meas},n}\\
&T(k+1)\leq T^\text{max}, \,\, T(0) = T_\text{meas}
\end{align}
\end{subequations}
for all $k=0,\hdots, K-1$ and $h=1,\hdots,H$. 

\subsection{Non-centralized Hub Formulation}
A similar decomposition from Case Study~1 can be used to form the partial Lagrangian,
\begin{align}
\mathcal{L}&(\{\mathbf{i}_h\}_{h=1}^H, \{\mathbf{E}_h,\}_{h=1}^H, \{i_m^\text{PW}\}_{m=1}^M,\mathbf{\lambda})\notag\\ 
&=\sum_{h=1}^H J_h(\mathbf{i}_h,\mathbf{E}_h)+\mathbf{\lambda}^\top\left(\mathbf i_d+\sum_{h=1}^H \mathbf{i}_h -\sum_{m=1}^M \mathbf i_m^\text{PW}
\right) \\
&=\sum_{h=1}^H\left( J_h(\mathbf{i}_h,\mathbf{E}_h)+\mathbf{\lambda}^\top \mathbf{i}_h \right)+\mathbf{\lambda}^\top(\mathbf i_d-\sum_{m=1}^M \mathbf i_m^\text{PW}
)).
\end{align}

Now, the primal problem formulation is in the same form as in Case Study~1 and we can use the same ALADIN, ADMM, and dual decomposition algorithms from Section~\ref{sec:distributed}.
\begin{table}
    \begin{center}
    \caption{Simulation parameters for case study~2}
    \scalebox{0.9}{
         \begin{tabular}{l c c}
         \toprule
         \textbf{Variable} & \textbf{Value} & \textbf{Units} \\
      \midrule
         \multicolumn{3}{c}{\textit{System parameters}} \\
         \midrule
         $H$ & 4 & Hubs \\
         $N_h$ & 100 & EVs \\
         $K$ & 240 & Timesteps\\
         $M$ & 15 & Segments \\
         $\Delta t$ & 180 & Seconds \\
         $ \tilde T_a(k)$ & $[16,18]$ & $^\circ$C \\
         $ c$ & 29.87 & - \\
         $i_d(k)$ & $[52,64]$ & kA \\
         $T^\text{max}$ & 100 & $^\circ$C \\
         $T_0$ & 70 & $^\circ$C \\
         $\gamma, \tau, \rho$ & \{0.000524, 0.9145, 0.0855\}  &  $\{^{^\circ C}/_{(\text{kA})^2},-,-\}$ \\
         $R$ & 480/13200 & \\
          \midrule
         \multicolumn{3}{c}{\textit{Commercial EV hub parameters}} \\
         \midrule
         $i_h^{\text{max}}$ & $[200,1000]$ & A \\
         $\alpha_n$ & $[80, 90]$ & \% \\
         $E_{h,n}^{\text{max}}$ & $[100, 600]$ & kWh \\
         $\underline s_{h,n}$ & $[10, 40]$ & \% \\
         $\bar s_{h,n}$ & $[80, 100]$  & \% \\
         $\underline{K}_{h,n}$ & [20:00, 03:00] & - \\
         $\bar K _{h,n}$ & [04:00, 07:00] & - \\
         $ q_h$, $ r_h$, $ o_h$ & $[0.1, 200]$, $10$, $100$ & -\\
         \bottomrule
        \end{tabular}
    }
    \label{table:simP_c2}
    \end{center}
\end{table}
\subsection{Simulation Setup for Case Study~2}
A hub model simulation was conducted for $H=4$ hubs with $N_h=100$ EVs in each hub. The distribution-level transformer in this scenario is a large 100MVA transformer  with a primary network voltage rating of $V_\text{pri} = 13.2kV$. Within the hubs, the secondary network supplies commercial chargers with RMS voltage at $V_\text{sec}=480V$. Since this scenario focuses on commercial vehicles, the battery capacities have been sized accordingly at 100, 200, or 600kWh with charging rates between 96-480kVA. For simplicity a constant background load of 25-30MVA is used. Table~\ref{table:simP_c2} presents relevant parameters for case study~2. 

\begin{figure}
    \centering
    \includegraphics[width=250pt]{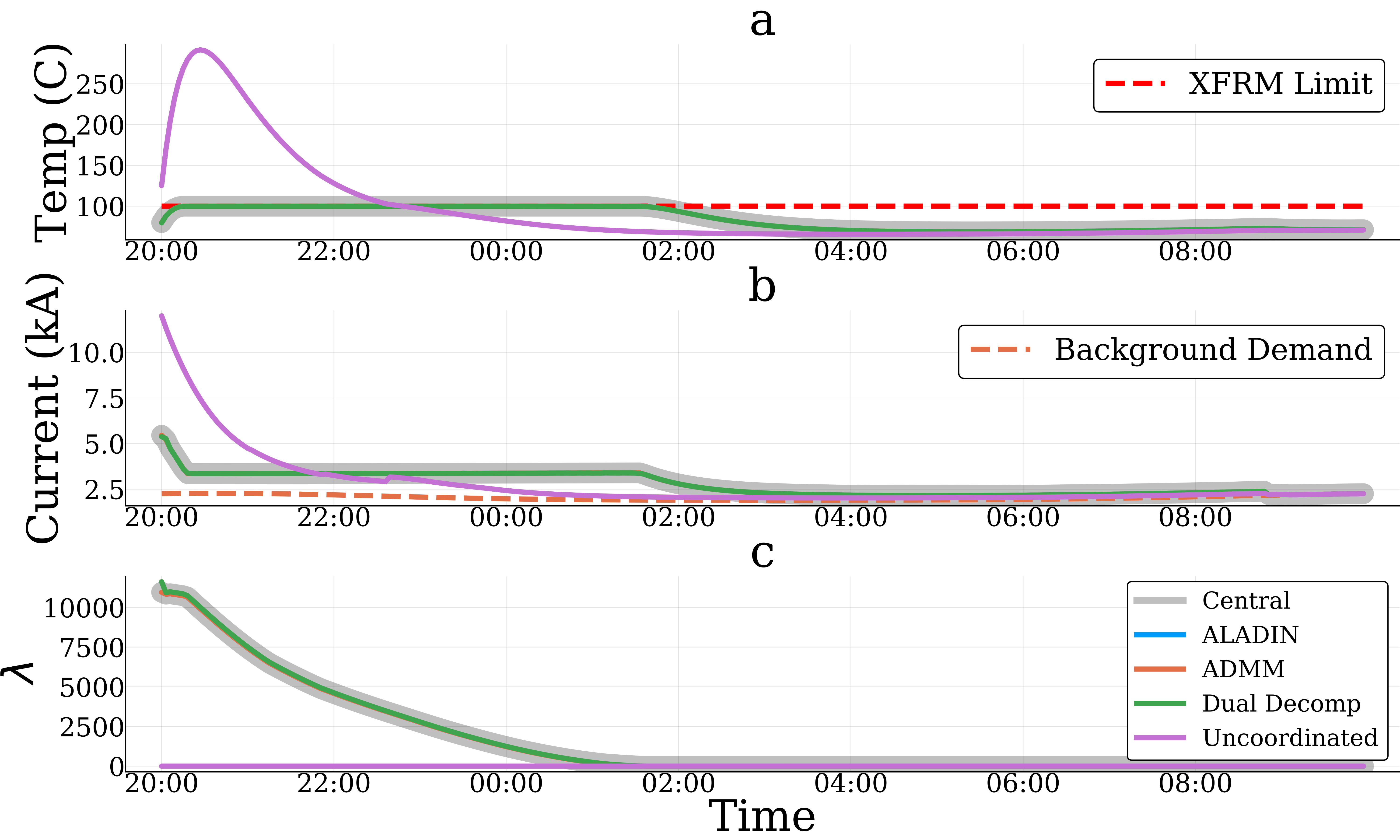}
    \caption{Case Study~2 results for allocating current capacity to hubs. \textit{Top}: Temperature response; \textit{Middle}: Total primary network current demanded from transformer and equal to reflected total secondary current ($R i_\text{total}(k)$); \textit{Bottom}: Dual variable.}
    \label{case2Results}
\end{figure}

\subsection{Discussion of Case Study~2 Results}
The central and uncoordinated results can be seen compared with the optimization algorithms solutions in Fig.~\ref{case2Results}. Once again, ALADIN and ADMM perform  well and match the central solution. Unlike in the residential scenario, PEM is not suitable for managing ON/OFF charging packets for hubs represents the equivalent of up to 100 EVs charging coincidentally, which begets large (bulky) demands on the transformer and makes temperature regulation challenging. Thus, PEM is not part of the commercial hub scenario, except as the possible intra-hub EV charging coordinator that ensures that a hub's aggregate demand is below its optimized static current capacity allocation. The comparison of the four distributed methods in case studies~1 and~2 is next. 

\section{Comparison of Distributed Methods} \label{sec:comp}
\subsection{Privacy, Performance, and Processing}
In Sections~\ref{sec:evc} and \ref{sec:hub}, the results of a receding-horizon implementation of the OCP is presented for each case study. In this section, we discuss how these methods performed and compare each method in terms of privacy (communication), performance (optimality), and processing (computation). 

\subsubsection{Privacy}
Table~\ref{infoTable} shows the information communicated between the EVs, Coordinator, and Transformer. The most valuable information from a consumer standpoint is the current and SoC schedules. While both Dual Ascent and ADMM transfer the current schedule to the coordinator, the coordinator only uses the sum of the current schedules so this sensitive information could be passed through a third party and aggregated first. However, in ALADIN, the individual current schedule is used in the coordinator problem as well as in the gradient. To approximate information requirements per timestep, consider an average number of iterations, a population of $N=100$ EV chargers, and a prediction horizon of 160 timesteps. Then, 
breaking these numbers into the data communicated per time step and per EV, we get 32~bits for PEM, 21~megabits for Dual Decomposition, 3~megabits for ADMM, and 0.6 megabits for ALADIN. That is, PEM and ALADIN require far less data to be communicated than the other two methods.
\begin{table}[htp!]
    \centering
    \caption{Summary of distributed methods - information sharing.}
    \resizebox{\columnwidth}{!}{
    \begin{tabular}{l c c c}
        \toprule
          \textbf{From-to} & \textbf{ALADIN} & \textbf{Dual Ascent / ADMM} & \textbf{PEM}\\
         \midrule
        \begin{tabular}{@{}r@{}}EV to \\ Coordinator \end{tabular} & \begin{tabular}{@{}c@{}}
             $\mathbf{i}_n^{(p)}$,  $\mathbf{g}_{i,n}^{(p)}$, $\mathbf{g}_{s,n}^{(p)}$, \\ $\mathbf{C}_{\overline{i},n}^{(p)}$, $\mathbf{C}_{\underline{i},n}^{(p)}$, $\mathbf{C}_{\overline{s},n}^{(p)}$, $\mathbf{C}_{\underline{s},n}^{(p)}$
        \end{tabular}&$\mathbf{i}_n^{(p)}$ & $\text{Req}_n(k)$\\ [0.5cm] 
        \midrule
        \begin{tabular}{@{}r@{}}Transformer to \\ Coordinator \end{tabular}  &\begin{tabular}{@{}c@{}}
             $\sum_{m=1}^M (\mathbf{i}_m^{PW})^{(p)}$,  $\mathbf{g}_{i,PW}^{(p)}$, \\ $\mathbf{g}_t^{(p)}$,  $\mathbf{C}_{\overline{z}}^{(p)}$, $\mathbf{C}_{\underline{z}}^{(p)}$, $\mathbf{C}_{\overline{t}}^{(p)}$, $\mathbf{C}_{\underline{t}}^{(p)}$
        \end{tabular}& $\sum_{m=1}^M (\mathbf{i}_m^{PW})^{(p)}$&$T(k)$ \\ [0.4cm] 
        \midrule
        \begin{tabular}{@{}r@{}}Coordinator to \\ EV \end{tabular}& $\mathbf\lambda^{(p)}$, $\mathbf{V}_{i,n}^{(p)}$, $\mathbf{V}_{s,n}^{(p)}$ & $\mathbf\lambda^{(p)}$, ($\mathbf{V}_{i,n}^{(p)}$)  & $\text{Resp}_n(k)$ \\ [0.2cm] 
        \midrule
        \begin{tabular}{@{}r@{}}Coordinator to \\ Transformer \end{tabular}  & $\mathbf\lambda^{(p)}$, $\mathbf{V}_t^{(p)}$, $\mathbf{V}_{i,PW}^{(p)}$& $\mathbf\lambda^{(p)}$, ($\mathbf{V}_{i,PW}^{(p)}$)   &  -\\
        \bottomrule
    \end{tabular}}
    \label{infoTable}
\end{table}

\subsubsection{Performance}
A summary of the performance of the four distributed methods is shown in Table~\ref{Tab:perfTable}. Specifically, it compares the 2-norm of the difference between the centralized method's optimal current schedules ($\mathbf{i}_n^\ast$) and dual variables ($\mathbf{\lambda}^\ast$) and the optimized values from the distributed methods. In both case studies, ALADIN and ADMM performed well as their solutions achieved optimality. Dual decomposition does not converge completely in the allotted time and performs worse as a result. The PEM coordinator focuses on feasibility of local and transformer problems with device-driven priorities and has no optimality guarantees; therefore, the difference in the current schedules are more pronounced for case study~1.

\begin{table}[htp!]
    \centering
    \caption{Summary of distributed methods - performance.}
    \scalebox{0.8}{
        \centering
         \begin{tabular}{l cccc}
              \toprule
            \multirow{2}[3]{*}{\textbf{Method} } & \multicolumn{2}{c}{\textbf{2-Norm Current Schedule}} & \multicolumn{2}{c}{\textbf{2-Norm Lambda}} \\
            \cmidrule(lr){2-3} \cmidrule(lr){4-5}
              & Case Study 1 & Case Study 2 & Case Study 1 & Case Study 2 \\ \midrule
             ALADIN &  $1\mathrm{e}{1}$ & $2\mathrm{e}{2}$& 
             $6\mathrm{e}{-4}$& $4\mathrm{e}{-3}$\\
             
             ADMM &  $8\mathrm{e}{1}$  & $3\mathrm{e}{2}$ & $4\mathrm{e}{-3}$ & $3\mathrm{e}{-2}$\\
              
             Dual Decomposition & $2\mathrm{e}{2}$  & $3\mathrm{e}{3}$ & $6\mathrm{e}{-2}$ & $8\mathrm{e}{-1}$\\
             
             PEM & $5\mathrm{e}{3}$  & - & -  & -\\
             \bottomrule
        \end{tabular}
    }
    \label{Tab:perfTable}
\end{table}

\subsubsection{Processing}
The computational efficiency of the methods can be seen in Table~\ref{Tab:compTimeTable}. The average solver time metric describes the average time it takes the algorithm to process for each time step. This number is not necessarily proportional to the average number of iterations shown in the second column as some algorithms require more processing per iteration. The PEM implementation requires least processing as it is an iteration-free approach. ALADIN is the next quickest followed by ADMM and Dual Decompostion. In the implementation, the algorithms have a constraint on the number of iterations due to the duration of each time step. Increasing the number of electric vehicles in the simulation would likely have a similar number of iterations per time step however the performance especially for dual decomposition and ADMM would decrease. It is worth noting that the stopping criteria was different for Case Study~1 and Case Study~2. In addition, the centralized results are only meant to be representative at the proposed scale as direct load control does not scale well in practice when the number of agents (EVs or hubs) or the prediction horizon increases.

\begin{table}[htp!]
    \centering
    \caption{Summary of distributed methods - processing.}
    \scalebox{0.8}{
        \centering
         \begin{tabular}{lcccc} 
            \toprule
             \multirow{2}[3]{*}{\textbf{Method} } & \multicolumn{2}{c}{\textbf{Average Solver Time/Iter. (Sec)}} & \multicolumn{2}{c}{\textbf{Average Iter. to Converge}} \\
            \cmidrule(lr){2-3} \cmidrule(lr){4-5}
              & Case Study 1 & Case Study 2 & Case Study 1 & Case Study 2 \\ \midrule
             Central &  3e-1 & 4e-2 & 1& 1\\
             ALADIN &  2e-1 & 4e-2 & 1.9& 1.4\\
             ADMM &  8e-3  & 1e-2 & 6.9& 18\\
             Dual Decomposition & 2e-3 & 5e-3 & 284& 1000\\
             PEM & 6e-2  & - & 1& -\\
             \bottomrule
        \end{tabular}
    }
    \label{Tab:compTimeTable}
\end{table}

\subsection{Summary of Results}
A qualitative summary of the differences in the distributed methods is shown in Fig.~\ref{spiderComp}. The central formulation gives the optimal solution quickly but gives no privacy and has a high communication overhead at scale. Dual decomposition and ADMM improve on data privacy, but see a significant decrease in the performance and computational efficiency. ALADIN shows the best performance out of the distributed methods but sacrifices privacy. PEM contrasts interestingly with Centralized and offers maximum privacy and speed but without optimality guarantees.

\begin{figure}[t]
    \centering
    \scalebox{0.6}{
        \begin{tikzpicture}
            \tkzKiviatDiagram[scale=0.75,label distance=2cm,radial=3, gap=1, lattice=5]{\vspace{4mm}Performance, Privacy, Processing}
            \tkzKiviatLine[thick,color=orange,fill=orange!20,opacity=.5](5,0,4) 
            \tkzKiviatLine[thick,color=red,fill=red!20](4.5,1,3) 
            \tkzKiviatLine[thick,color=blue,fill=blue!20,opacity=.5](1,4.5,4.5)  
            \tkzKiviatLine[thick,color=brown,fill=brown!20,opacity=.5](2,3.5,0.5)  
            \tkzKiviatLine[thick,color=green,fill=green!20,opacity=.5](3,3,1.5)   
            \node[anchor=south west,xshift=-120pt,yshift=20pt] at (current bounding box.south east) 
                {
                \begin{tabular}{@{}lp{3cm}@{}}
                \ColorBox{orange} & Central \\
                \ColorBox{red} & ALADIN \\
                \ColorBox{green} & ADMM \\
                \ColorBox{brown} & Dual Decomp \\
                \ColorBox{blue} & PEM \\
                \end{tabular}
                };
        \end{tikzpicture}
        }
    \caption{Qualitative relative ranking of the different EVC control methods.}
    \label{spiderComp}
\end{figure}
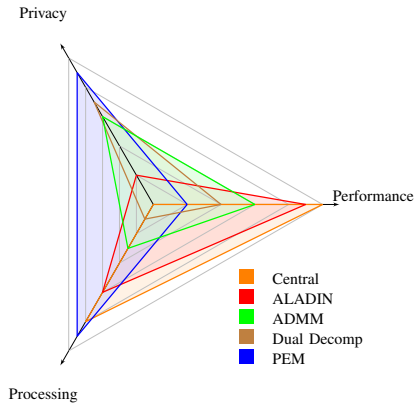

\subsection{Selecting a Suitable Distributed Method}

For the scenario in case study 1, privacy is important as residential EV owners should not need to share their private driving information. Using ALADIN, the coordinator knows the gradients which are a scaled version of the current schedule, which is sensitive data. Due to the large amount of information being shared with the ALADIN algorithm, this may not be the best approach even though it shows the best performance.  Thus, ADMM or PEM are well-suite for residential fleets, such as in case study~1. For commercial fleets, such as in case study~2, where data privacy is less of a priority, ALADIN is a powerful option. Clearly, for a general problem, the order of priorities must be decided before deciding on a specific method.

\section{Conclusion \& Future Work} \label{sec:conc}
Utilities and other entities in the energy industry will soon have to consider the impacts of increased adoption of electric vehicles. Uncoordinated charging could cause overloading of grid transformers as the penetration of electric vehicles increases. We have shown that there are multiple distributed control strategies that could be implemented to avoid costly upgrades of these devices. We have compared the tradeoff of these methods in the areas of privacy (communication), performance (optimality), and processing (computation time). 

Based on two different case studies, we found that 
the proposed and novel suboptimal, but privacy-preserving algorithm PEM might be ideal for an application where privacy is valued, such as a residential EV charging. On the other hand, in a commercial setting, such as the hub charging problem, where performance may be priority, ALADIN is a good choice. 

Since future control methods will require more than EVs to be coordinated and incorporate AC grid reliability constraints, access to and the role of (private) information will be critical. Within this context, we are interested to extend ALADIN and PEM to consider other types of distributed energy resources (DERs) for demand management and adding additional grid constraints via AC optimal power flow problems and investigate the effects of negative background demand due to solar PV and vehicle-to-grid capabilities. By incorporating utility network information into the control algorithms, it also becomes critical to consider the role of cybersecurity, which is an important aspect of data privacy.

\section{Appendix - Proof of Theorem~\ref{thm1}}
The tightness proof for the relaxed transformer dynamics in Theorem~\ref{thm1} relies on KKT analysis. Thus, we first need to define the primal constraints and dual variables of the relevant SOCP formulation of~\eqref{eq:OCP} along with the KKT stationarity conditions. These a presented below before the proof.

\subsection{Primal constraints and dual variables}
Consider the primal SOCP constraints from~\eqref{eq:OCP} with the state of charge limits removed as they are not needed for the conditions in Theorem~\ref{thm1}.
For all $k \in \mathcal{K} \doteq \{0, \hdots, K-1\}$ and $n\in \mathcal{N} \doteq \{1,\hdots, N\}$, the following constraints define primal feasibility and dual variables after  $|$:
{\small 
\begin{subequations}
\begin{align}
0 &= T(k+1) - \tau T(k)  - \eta e(k)   & |\, \lambda_T^{k+1}  \in \mathbb{R}    \,\, \forall k \\
0 &= s_n(k+1) - s_n(k) - \eta_n i_n(k)   & |\, \lambda_{s_n}^{k+1} \in \mathbb{R}   \,\, \forall k, \,\, \forall n \label{eq:SoCDyn}\\
0 &= i_\text{total}(k) - i_d(k) -  \sum_{n=1}^N  i_n(k)  & |\, \lambda_{c}^{k} \in \mathbb{R}   \,\, \forall k\\
0 &\ge T(k+1) - T^\text{max}  & |\, \mu_{T}^{k+1} \ge 0   \,\, \forall k \label{eq:maxTemp} \\
0 &\ge (i_\text{total}(k))^2 - e(k)   & |\, \mu_{e}^{k} \ge 0   \,\, \forall k \label{eq:cvxRelax}\\
0 &\ge i_n(k) - i_n^\text{max}   & |\, \overline{\mu}_{i_n}^{k}  \ge 0   \,\, \forall k, \,\, \forall n \\
0 &\ge 0 - i_n(k)    & |\, \underline{\mu}_{i_n}^{k} \ge 0   \,\, \forall k, \,\, \forall n \\
 0& \ge \bar s_n- s_n(\bar k_n+1) ,   & |\, \bar {\mu}_{s_n}^{\bar{k}_n+1}  \ge 0   \,\, \forall n
\end{align}
\end{subequations}
}
where the last constraint is the QoS guarantee that ensures that vehicle $n$ achieves at least a state-of-charge of $\bar s_n$ by no later than time $\bar k_n+1$. 
Without loss of generality, we can also set $i_d(k) \equiv 0$ and assume $s_n(0)>0$.

\subsection{KKT stationarity conditions}
If we assume Slater's constraint qualification holds\footnote{This is reasonable for the SOCP formulation and equivalent to the existence of a strictly feasible solution where the transformer temperature is not at its limit at all times, i.e., we have some flexibility in the system.}, the stationarity condition $\nabla_{x(k)} \mathcal{L}(x,\lambda, \mu) =0$ has to hold for each variable $x$ at timestep $k$, which gives
{\small
\begin{subequations}
\begin{align}
\nabla_{T(k+1)} \mathcal{L} & \Rightarrow  \lambda_T^{k+1} = \tau \lambda_T^{k+2} - \mu_T^{k+1} \label{eq:Lag_temp}\\
\nabla_{T(K)} \mathcal{L} & \Rightarrow \lambda_T^{K} = - \mu_T^{K} \label{eq:Lag_tempK}\\
\nabla_{e(k)} \mathcal{L} &\Rightarrow 0 =  - \eta \lambda_{T}^{k+1} - \mu_e^{k} \label{eq:Lag_e}\\
\nabla_{i_\text{total}(k)} \mathcal{L} &\Rightarrow 0 =  \lambda_{c}^{k} + 2\mu_e^k i_\text{total}(k) \label{eq:Lag_total}\\ 
\nabla_{i_n(k)} \mathcal{L}&\Rightarrow 0 = 2r_n i_n(k) - \eta_n \lambda_{s_n}^{k+1} - \lambda_{c}^k + \overline{\mu}_{i_n}^{k}-\underline{\mu}_{i_n}^{k} \label{eq:Lag_in}\\
\nabla_{s_n(k+1)} \mathcal{L} &\Rightarrow  \lambda_{s_n}^{k+1} = \lambda_{s_n}^{k+2} + 2q_n(1-s_n(k+1))  
\label{eq:Lag_sn}\\
\nabla_{s_n(\bar k_n+1)} \mathcal{L} &\Rightarrow  \lambda_{s_n}^{\bar k_n+1} = \lambda_{s_n}^{\bar k_n+2}  + 2q_n(1-s_n(\bar k_n+1)) 
+ \bar \mu_{s_n}^{\bar k_n+1} \label{eq:Lag_snKbar}\\
\nabla_{s_n(K)} \mathcal{L} &\Rightarrow \lambda_{s_n}^{K}  = 2q_n(1-s_n(K)). 
\label{eq:Lag_snK}
\end{align}
\end{subequations}
}

Before we can complete the proof, we need help from three technical lemmas that employ the primal and dual relations.

\begin{lemma}\label{lemma:dualTemp}
At optimality, the dual variable, $\mu_e^k$, associated with relaxed quadratic constraint~\eqref{eq:cvxRelax}, satisfies $\mu_e^l \ge \mu_e^k$ for all $l\le k$.
Specifically, if~\eqref{eq:cvxRelax} is strictly active at timestep $k$, then it is strictly active for all prior timesteps. 
\end{lemma}
\begin{proof}
From recursion on \eqref{eq:Lag_temp} and \eqref{eq:Lag_tempK}, it is trivial to show that for all $k$
\begin{align} \label{eq:Trec}
{\lambda_T^{k+1} = - \sum_{t=k+1}^K \tau^{t-k-1} \mu_T^{t},} 
 \end{align}
where $\mu_T^{k+1}\ge 0\,\, \forall k\in\mathcal{K}$. Substituting $\lambda_T^{k+1}$ from~\eqref{eq:Trec} into~\eqref{eq:Lag_e}, we have that for all $l\le k$
\begin{align} 
\mu_e^l	&= \eta \sum_{t=l+1}^K \tau^{t-l-1} \mu_T^{t} \label{eq:lem1_mue}\\
        &= \eta \left( \sum_{t=l+1}^k \tau^{t-l-1} \mu_T^{t} +  \sum_{t=k+1}^K \tau^{t-k-1} \mu_T^{t}\right) \ge  \mu_e^k.
\end{align}
 Thus, if $\mu_e^k > 0 \Rightarrow  \mu_e^l > 0 \,\, \forall l\le k$, concluding 
 the proof.
\end{proof}

\begin{lemma}\label{lemma:dualSn}
 From \eqref{eq:Lag_sn} and \eqref{eq:Lag_snK} and recursion on $k$, it is trivial to show that for all $k$, 
\begin{align} \label{eq:sn_rec}
{\lambda_{s_n}^{k+1} = \sum_{t=k+1}^K 2q_n(1-s_n(t)) + I_{\bar k, k} \bar \mu_{s_n}^{\bar k_n +1},}
\end{align}
where $I_{\bar k, k} = 1$ if $k \le \bar k_n$ and $ I_{\bar k, k} =0$ otherwise.
\end{lemma}

\begin{lemma}\label{lemma:SoCseq}
Since $i_n(k)\ge 0$, the sequence $\{s_n(k)\}_{k=1}^{K}$ defined by \eqref{eq:SoCDyn} is clearly non-decreasing for all $n\in \mathcal{N}$. That is, $1\ge s_n(k+1)\ge s_n(k) \ge s_n(0)$ for all $k\in\mathcal{K}$.
\end{lemma}



\subsection{Proof of Theorem~\ref{thm1}}

\begin{proof}(Direct) From Lemma~\ref{lemma:dualTemp}, we just need to show that $\mu_e^k>0$. Thus, first consider~\eqref{eq:Lag_in} for timestep $k$ and substitute for $\lambda_{c}^k$ with \eqref{eq:Lag_total} and $\lambda_{s_n}^{k+1}$ with Lemma~\ref{lemma:dualSn}, which yields 
\begin{align*}
0 &= 2r_n i_n(k) - \eta_n \left( 2q_n\sum_{t=k+1}^K (1-s_n(t)) + I_{\bar k, k} \bar \mu_{s_n}^{\bar k_n +1}\right) \\
& \quad + 2\mu_e^k i_\text{total}(k) + \overline{\mu}_{i_n}^{k} - \underline{\mu}_{i_n}^{k}.
\end{align*}
and $\mu_e^k$ to the left-hand side gives
\begin{align}
 2 \mu_e^k i_\text{total}(k) &= \eta_n \left( 2q_n\sum_{t=k+1}^K (1-s_n(t)) + I_{\bar k, k} \bar \mu_{s_n}^{\bar k_n +1}\right) \notag \\
 & \quad + \underline{\mu}_{i_n}^{k} - 2r_n i_n(k) -  \overline{\mu}_{i_n}^{k}. \label{eq:proofMidstep}
\end{align}
Since the transformer is overloaded due to excessive demand, $i_\text{total}(l) > 0$. Thus, we just need to show that RHS is strictly positive. Before doing so, we first simplify the notation by defining $\alpha_n(k,\bar k_n) \doteq\eta_n I_{\bar k, k} \bar \mu_{s_n}^{\bar k_n +1} + \underline{\mu}_{i_n}^{k} - \overline{\mu}_{i_n}^{k} $ where $\alpha_n(k,\bar k_n)\ge 0$ since $i_n(k)<i_n^\text{max}$.
Clearly, if $r_n=0$, the proof is complete for $s_n(k+1)<1$. However, for $r_n>0$, we need to consider the ratio $q_n/r_n$. Thus, we will use~\eqref{eq:SoCDyn} and replace  $i_n(k)$ with $\frac{1}{\eta_n}(s_n(k+1)-s_n(k))$ in~\eqref{eq:proofMidstep}:
 \begin{align*}
 2 \mu_e^k i_\text{total}(k) 	
					&   = 2q_n \eta_n \sum_{t=k+1}^K (1-s_n(t)) 
					\\ & \quad - 2\frac{r_n}{\eta_n}\left(s_n(k+1) - s_n(k)\right)  + \alpha_n(k, \bar k_n) \notag \\
					& \ge 2q_n \eta_n \left( (1-s_n(k+1))  + \sum_{t=k+2}^K (1-s_n(t))\right)  \\
					& \quad - 2\frac{r_n}{\eta_n}\left(s_n(k+1) - s_n(k)\right),
\end{align*}
where the inequality is due to $\alpha_n(k,\bar k_n) \ge 0$. Further reductions show that $2 \mu_e^k i_\text{total}(k)$
\begin{align*}
&   \ge 2q_n \eta_n (1-s_n(k+1)) - 2\frac{r_n}{\eta_n}\left(s_n(k+1) - s_n(k)\right)  \notag \\
&   \ge 2\frac{r_n}{\eta_n} \left( \frac{q_n}{r_n} \eta_n^2 (1-s_n(k+1)) - s_n(k+1) + s_n(0) \right) \notag \\
&   = 2\frac{r_n}{\eta_n} \left( M_n + s_n(0) - (M_n+1) s_n(k+1) \right) 
\end{align*}
where the last inequality is due to Lemma~\ref{lemma:SoCseq}. For $r_n>0$ and $s_n(k+1) < \frac{M_n + s_n(0)}{M_n+1}$, the RHS is strictly positive, which ensures that $ \mu_e^k > 0$.
Finally, from Lemma~\ref{lemma:dualTemp}, we have $\mu_e^l \ge \mu_e^k > 0 \,\,\forall l\le k$, which completes the proof.

\end{proof}

\section{Appendix - Proof of Corollary~\ref{corol1}}
The corollary refers to strictly active inequality constraints~\eqref{eq:maxTemp} and~\eqref{eq:cvxRelax}, which relates to dual conditions $\mu_T^{k+1}>0$ and $\mu_e^k>0$, respectively.
\begin{proof}
This proof has two parts.
\begin{itemize}
\item (Proving if $k$ is largest timestep to satisfy $\mu_e^k > 0$ then $\mu_T^{k+1}>0$ and $\mu_T^{m}=0\, \forall m>k+1$ ):
Recall that in~\eqref{eq:lem1_mue}, since $\eta, \tau > 0$ and $\mu_T^t\ge 0$, if $\mu_e^k>0$, then $\exists m > k$ such that $\mu_T^{m+1} > 0$. Now, assume $m > k+1$, then 
$$\mu_e^{m} = \eta \sum_{t=m+1}^K \tau^{t-m-1} \mu_T^{t} \Rightarrow \mu_e^{m} > 0.$$
However, since $m>k$ that contradicts with $k$ being the largest integer for which $\mu_e^k > 0$ and, thus,  $\mu_T^{m} =0 \,\,\forall m > k+1$ and $k+1$ is the last instance of $\mu_T^{k+1} >  0$, which implies that $T(k+1)=T^\text{max}$.

\item (Proving if $k+1$ is last timestep with $\mu_T^{k+1} > 0$ then  $\mu_e^k > 0$): if $k+1$ is the last instance of $\mu_T^{k+1} > 0$ then $\mu_e^l > 0 \,\, \forall l\le k$ and, thus, $k$ is the largest integer for which $\mu_e^k>0$ and $e(k) = (i_\text{total}(k))^2$. 
\end{itemize}
This completes the proof.
\end{proof}

Note that in a practical setting, where optimality of EV charging control is not critical, a practitioner could circumvent the complexity of the convex relaxation by augmenting objective function~\eqref{eq:objFun} with a temperature deviation term $-\epsilon (T^\text{max} - T(k+1))$ for arbitrarily small $\epsilon>0$. This incentivizes temperature trajectories far from the temperature limit by embedding a $-\epsilon$ into the RHS of~\eqref{eq:Lag_temp} and~\eqref{eq:Lag_tempK}, which guarantees that $\lambda_T^{k+1}<0\,\, \forall k\in\mathcal{K}$.  From~\eqref{eq:Lag_e}, this yields $\mu_\text{e}^k > 0\,\, \forall k$, which ensures that the convex relaxation is tight for all time steps, regardless of transformer or fleet conditions, and $q_n/r_n$ ratios. The practical impact of using this approach is that for larger $\epsilon>0$, the EV optimal charging schedule embodies a utility-centric, valley-filling policy~\cite{gametheory2,singleControl2}, which competes with that of the QoS-centric objective in~\eqref{eq:objFun} and may negatively impact EV customer satisfaction.

\bibliographystyle{IEEEtran}
\bibliography{mybib}






\end{document}